\RequirePackage{fix-cm}
\documentclass[smallextended]{svjour3}       
\smartqed  
\usepackage{graphicx}
\usepackage{epstopdf}
\usepackage{geometry}
\usepackage{graphicx,amssymb,amstext,amsmath}
\usepackage[justification=centering]{caption}

\usepackage{epsfig}
\usepackage{subfigure}
\usepackage{amsmath}
\usepackage{algorithmic}
\usepackage{graphics}
\usepackage{amssymb}
\usepackage{alltt}
\usepackage{float}
\usepackage{array}
\usepackage{xspace}
\usepackage{footmisc}
\usepackage{comment}
\usepackage{algorithm}
\usepackage{algorithmic}
\usepackage[misc]{ifsym}

\usepackage{multirow}
\usepackage{colortbl}
\usepackage{url}

\newcounter{RomanNumber}

\newcommand{\vig}{VIG-Tree\xspace}

\begin{document}

\title{Efficient Continuous Top-$k$ Geo-Image Search on Road Network
}
\subtitle{}


\author{Chengyuan Zhang  \and
        Kesheng Cheng \and
        Lei Zhu \and
        Ruipeng Chen \and
        Zuping Zhang \and
        Fang Huang \and
}


\institute{Chengyuan Zhang \at
            \email{cyzhang@csu.edu.cn}
            \and
            Kesheng Cheng \at
              \email{chengks0901@csu.edu.cn}
           \and
           Lei Zhu \at
              \email{leizhu@csu.edu.cn}
           \and
           Ruipeng Chen \at
              \email{rpchen@csu.edu.cn}
           \and
           \Letter Zuping Zhang \at
              \email{zpzhang@csu.edu.cn}
           \and
           Fang Huang \at
              \email{hfang@csu.edu.cn}
           \\
           \\
           School of Information Science and Engineering, Central South University, PR China
}

\date{Received: date / Accepted: date}

\maketitle

\begin{abstract}
With the rapid development of mobile Internet and cloud computing technology, large-scale multimedia data, e.g., texts, images, audio and videos have been generated, collected, stored and shared. In this paper, we propose a novel query problem named continuous top-$k$ geo-image query on road network which aims to search out a set of geo-visual objects based on road network distance proximity and visual content similarity. Existing approaches for spatial textual query and geo-image query cannot address this problem effectively because they do not consider both of visual content similarity and road network distance proximity on road network. In order to address this challenge effectively and efficiently, firstly we propose the definition of geo-visual objects and continuous top-$k$ geo-visual objects query on road network, then develop a score function for search. To improve the query efficiency in a large-scale road network, we propose the search algorithm named geo-visual search on road network based on a novel hybrid indexing framework called VIG-Tree, which combines G-Tree and visual inverted index technique. In addition, an important notion named safe interval and results updating rule are proposed, and based on them we develop an efficient algorithm named moving monitor algorithm to solve continuous query. Experimental evaluation on real multimedia dataset and road network dataset illustrates that our solution outperforms state-of-the-art method.
\end{abstract}

\section{Introduction}
\label{Intro}
With the rapid development of mobile Internet and cloud computing technology, large scale multimedia data~\cite{DBLP:conf/cikm/WangLZ13,DBLP:conf/mm/WangLWZZ14,DBLP:conf/mm/WangLWZ15}, e.g., texts, images~\cite{DBLP:journals/tip/WangLWZ17}, audio and videos have been generated, collected, stored and shared. For example. Facebook, the most famous online social networks services, reports more then 300 million photos uploaded and shared daily in the November 2013. More than 3.5 million photos had been uploaded by 87 million registered users of Flickr, which is the largest online photo sharing service. More than 140 million Twitter users posts 400 million tweets which contain 140 characters in text and images with geographical information like latitude and longitude. YouTube, the largest video sharing web site in the world, shares more than 100 hours of videos every minutes in the end of 2013. The total amount of users in Himalaya which is a popular audio sharing platform are already more than 470 million. As of December 2015, the total amount of audio has exceeded 15 million in Himalaya. Beyond all doubt, unlike the situation in the past, the mess multimedia data account for nearly 80\% total amount of data in the big data environment. As it knows to all, advanced mobile devices equipped with wireless network module, high definition camera and microphone such as smartphones and tablets, together with other popular mobile applications and location based services (LBS for short) like WeChat, Uber, Amap and etc., bring a lot of convenience to people everyday by collecting and sharing massive multimedia data with geo-location information. In the meanwhile, more challenges are raised concerning multimedia data retrieval~\cite{DBLP:journals/tnn/WangZWLZ17}.

As a significant problem in the area of multimedia retrieval~\cite{DBLP:conf/sigir/WangLWZZ15,DBLP:journals/tip/WangLWZZH15,DBLP:conf/ijcai/WangZWLFP16}, content-based image retrieval~\cite{DBLP:journals/cviu/WuWGHL18,DBLP:journals/pr/WuWLG18} is applied in many applications such as image processing and retrieval~\cite{DBLP:journals/corr/abs-1708-02288,DBLP:conf/pakdd/WangLZW14}. The basic objective of CBIR is to automatically extract low-level visual features from images, such as color, texture, edge, shape and etc. There are two important part of CBIR research, i.e., image representation and visual similarity measurement. Many researchers work for these two problem and many approaches have been proposed. Thomee et al.~\cite{DBLP:journals/ijmir/ThomeeL12} proposed a review of interactive search in image retrieval. Fu et al.~\cite{DBLP:conf/ICCC/Fu2016} proposed a solution which combines convolution neural network~\cite{DBLP:journals/pr/WuWGL18} (CNN) and CBIR for image retrieval. They utilized support victor machine (SVM) to train a hyerplane which can separate similar image pairs and dissimilar image pairs to a large degree. Norouzi et al.~\cite{DBLP:conf/nips/0002FS12} designed a mapping learning approach for large scale multimedia applications from high-dimensional data to binary codes that preserve semantic similarity.

The smart devices equipped high definition camera, wireless mobile communication module and GPS-module, like smartphone and tablet, are used by lots of users to take photos and shared on the Internet easily everyday and everywhere. These photos record nice sense and on the other hand, they have geo-tags which record the geographical information of where they were taken. We can treat these images with geo-tags as geo-visual objects which contains two parts of features, i.e., visual content features and geographical features. Like the problem of spatial textual query on road network, more and more people are beginning to pay close attention to geo-tagged image search on road network.

To the best of our knowledge, we are the first to study the problem of continuous top-$k$ geo-visual objects query on road network, which aims to search out $k$ best geo-visual objects generated from images with geo-tags taking into account two types of relevancy: (1)road network distance proximity between query location and objects, (2)visual content similarity between query and objects. Both of query and objects consist of geographical information and visual content information, and with the moving of query on road network, the query continuously returns satisfactory results. In order to describe this new problem more clearly and concretely, we introduce two examples as follows:

\begin{example}
\label{ex:example1}
As illustrated in Figure.~\ref{fig:fig1}, an user is driving back from work and she want to buy a handbag which is same as one in a photo. But she do not know its brand and the item number. Obviously, describing a handbag in great detail by a few words or a sentence is very challenging. Furthermore, she have no idea which nearby shop has the satisfactory style of handbags, she want to search out which shop has this type of handbag on her way home. In such case, she can input this photo and her current location information by her smartphone as a top-$k$ geo-visual object query on road network to a geo-tagged multimedia data retrieval system when her driving. Then this system according to the driving route returns a set containing $k$ geo-visual objects meeting her requirements. that can tell the user which shops have this kind of handbag and are close to her location. With the location of the user changing, the system will dynamically update the result set. In Figure.~\ref{fig:fig1}, the geo-visual objects are the cyan small point and the query point is in the red circle and the red arrow indicates it moving direction.
\end{example}

\begin{figure}[thb]
\newskip\subfigtoppskip \subfigtopskip = -0.1cm
\centering
\includegraphics[width=1\linewidth]{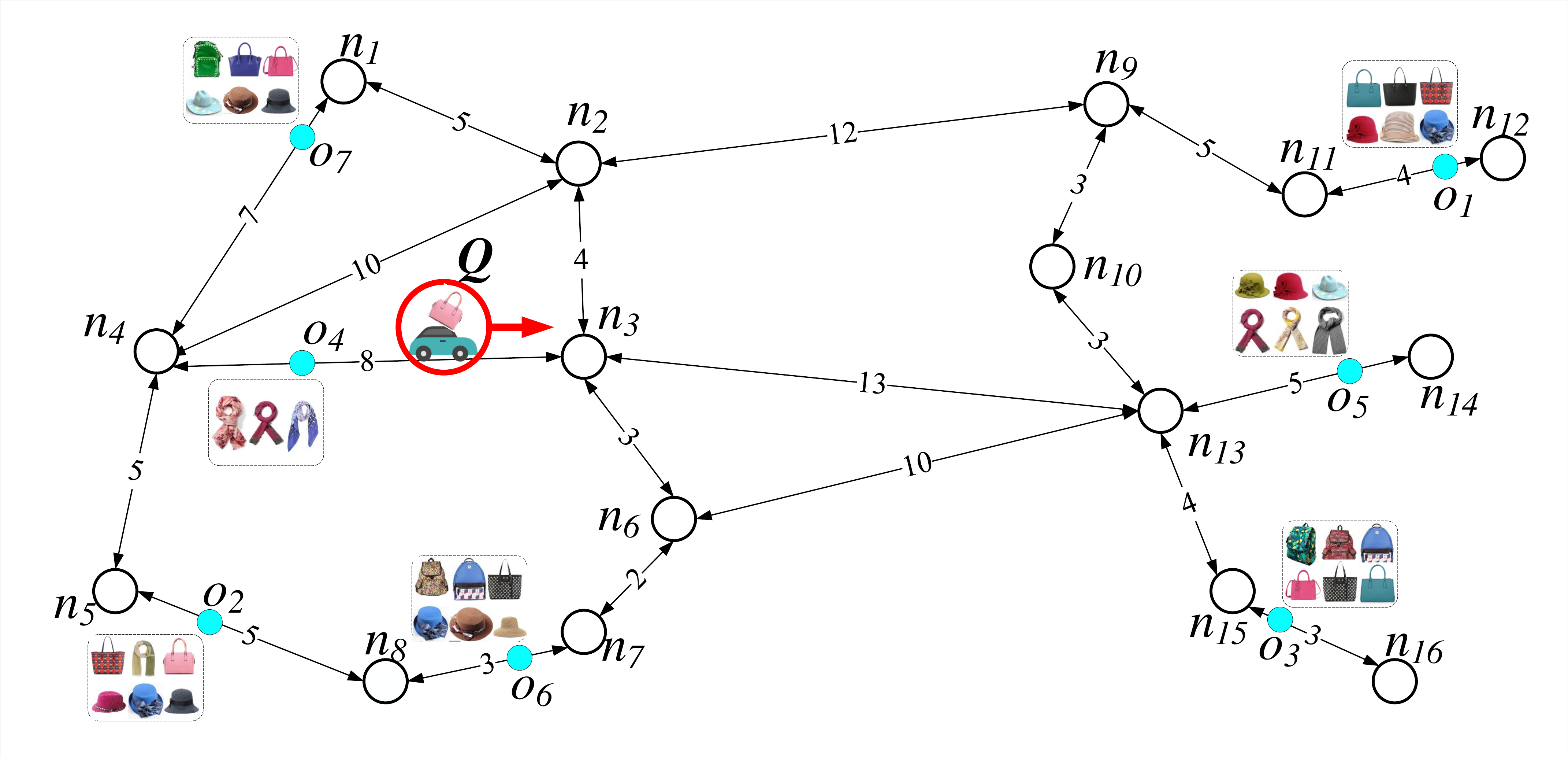}
\vspace{-1mm}
\caption{\small  An example of continuous top-$k$ geo-visual objects query for shopping }
\label{fig:fig1}
\end{figure}

\begin{figure}[thb]
\newskip\subfigtoppskip \subfigtopskip = -0.1cm
\centering
\includegraphics[width=1\linewidth]{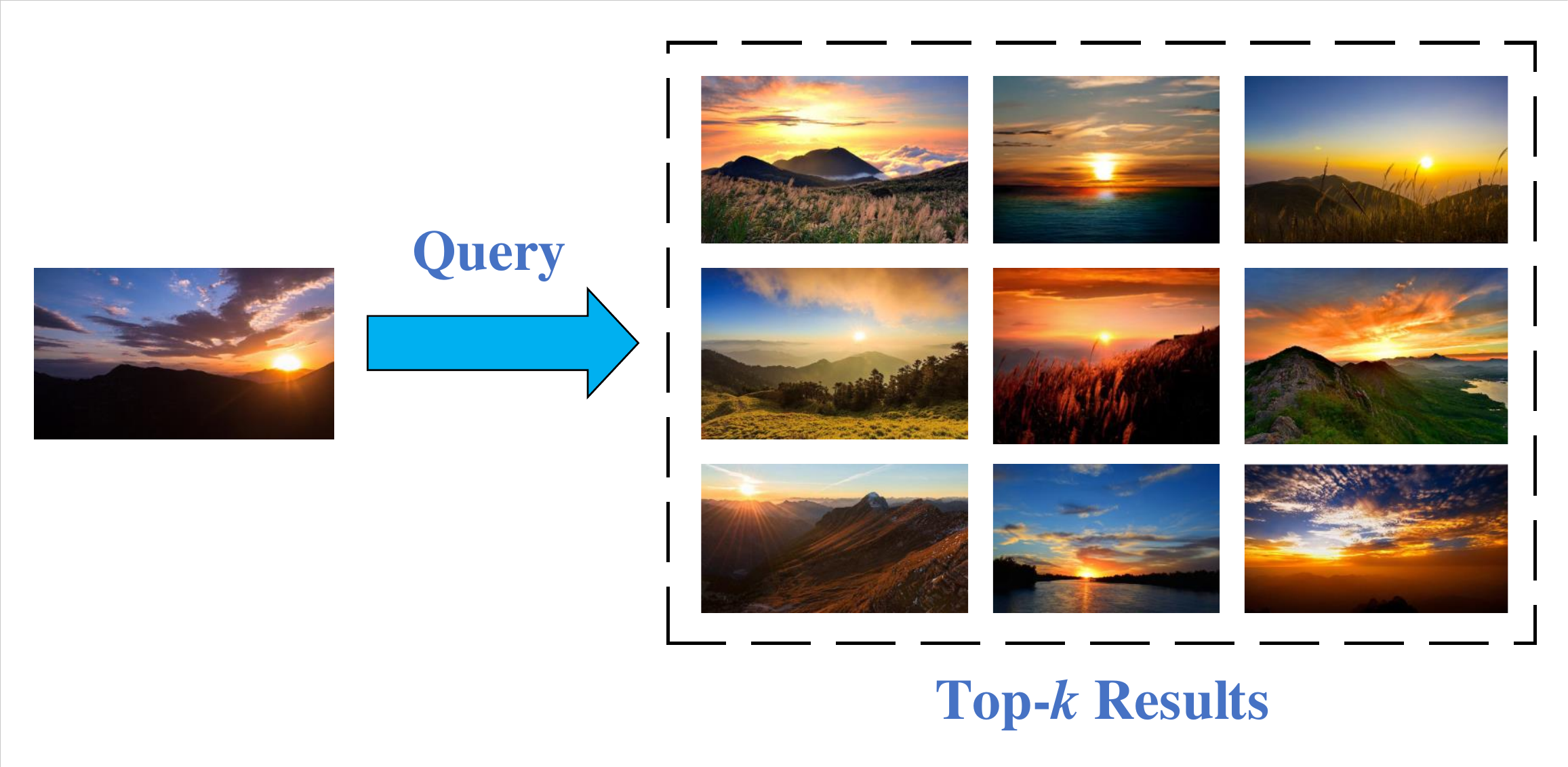}
\vspace{-1mm}
\caption{\small  An example of continuous top-$k$ geo-visual objects query for photographing }
\label{fig:fig2}
\end{figure}

\begin{example}
\label{ex:example2}
As illustrated in Figure.~\ref{fig:fig2}, a photographer traveling in a natural scenic area want to take some photo about sunset. As she is unfamiliar with the environment of this scenic area, she cannot choose a good position near her current location to take a good work. In this case, she do not need to depict the picture in her mind in detail in words. It is only necessary to select a photo which can represent her intention as the query image and then input it and her location to the geo-tagged multimedia retrieval system when she walks through the scenery. A top-$k$ geo-visual objects query will be processed and the set of results containing $k$ photos taking by other photographers or tourists can instruct her which position is the better choice and which route is the nearest.
\end{example}

This paper aims to implement the challenging application described in example~\ref{ex:example1} and example~\ref{ex:example2}, namely, implementing efficient continuous top-$k$ geo-visual objects query on road network. We propose a novel hybrid indexing framework named \vig which combines G-Tree and visual inverted index technique. To process geo-visual queries with VIG-Tree, we exploit the best-first traversal algorithm for retrieving the top-$k$ geo-visual objects. In order to further reduce the computational cost, we introduce the notion of safe interval in road network and design the results updating rule in the process of query moving. An efficient algorithm named moving monitor algorithm is developed to improve search efficiency.

\noindent\textbf{Contributions.}
Our main contributions can be summarized as follows:
\begin{itemize}
\item To the best of our knowledge, we are the first to study continuous top-$k$ geo-visual objects query on road network. Firstly we propose the definition of geo-visual object and top-$k$ geo-visual object query on road network. The socre function containing road network distance proximity component and visual content similarity component is designed and then we define the continuous geo-visual object and top-$k$ geo-visual object query on road network.
\item We present a hybrid index framework named VIG-Tree to support the geo-visual object search on road network, which combines G-Tree and visual inverted index technique. Then a algorithm named geo-visual search on road network is developed.
\item In order to reduce the computational cost of query, we propose the notion of safe interval and results updating rule, and then introduce moving monitor algorithm.
\item We have conducted extensive experiments on real multimedia dataset and road network dataset. Experimental results demonstrate that our solution outperforms the state-of-the-art method.
\end{itemize}

\noindent\textbf{Roadmap.} The remainder of this paper is organized as follows: We review the related work in Section 2. Section 3 introduces the definition of the continuous top-$k$ geo-visual objects query on road network as well as relevant notions. We propose a novel hybrid index framework named VIG-Tree and then present the geo-visual search on road network algorithm in section 4. Section 5 present the definition of safe interval and moving monitor algorithm. The experimental results are shown in Section 6, and finally we draw our conclusion of this paper in Section 7.

\section{Related work}
\label{related}
In this section, we introduce an overview of existing researches of image retrieval and spatial keyword query on road network, which are related to this work. To the best of our knowledge, there is no existing work on the problem of continuous top-$k$ geo-visual objects query on road network.

\subsection{Content-Based Image Retrieval}
Content-based image retrieval (CBIR for short) is important research problem in the area of multimedia system and retrieval~\cite{NNLS2018,DBLP:journals/ivc/WuW17,TC2018,DBLP:conf/mm/WuWS13}. Over the last decade, lots of researchers focus on this hot issue and made great progress. The Scale Invariant Feature Transform (SIFT for short) proposed by Lowe~\cite{DBLP:conf/iccv/Lowe99} is a popular approach which transforms an image into a large collection of local feature vectors. These local features are invariant to image translation, scaling, rotation, and partially invariant to illumination changes and affine or 3-dimension projection. This method consists of four stages~\cite{DBLP:journals/ijcv/Lowe04}: (1)scale-space peak selection; (2)keypoint localization; (3)orientation assignment; (4)keypoint descriptor. Based on SIFT, Ke et al.~\cite{DBLP:conf/cvpr/KeS04} proposed a descriptors named PCA-SIFT encode the salient aspects of the image gradient in the feature point's neighborhood and they used Principal Components Analysis (PCA) to the normalized gradient patch. Mortensen et al.~\cite{DBLP:conf/cvpr/MortensenDS05} introduced a feature descriptor which augments SIFT with a global context vector that adds curvilinear shape information from a much larger neighborhood. This descriptor is robust to local appearance ambiguity and non-rigid transformations. Liu et al.~\cite{DBLP:journals/inffus/LiuLW15} proposed a novel image fusion method for multi-focus images by applying dense SIFT. In order to improving the object retrieval performance, Zhang et al.~\cite{DBLP:journals/ijon/ZhangZZZW17} presented a novel method to employ CNN evidences to improve the SIFT matching accuracy.

Many other techniques have been designed to improve effectiveness and efficiency of content-based image retrieval. For content-based landmark image search problem, Zhu et al.~\cite{DBLP:journals/tcyb/ZhuSJZX15} proposed multimodal hypergraph (MMHG) to characterize the complex associations between landmark images. Furthermore, they developed a novel content-based visual landmark search system based on MMHG to improve effectiveness of searching. Xiao et al.~\cite{DBLP:journals/mta/XiaoQ14} presented a complementary relevance feedback-based content-based image retrieval system using short-term and long-term learning techniques to improve the retrieval performance. Jaffar et al.~\cite{} proposed a semantic image retrieval system in a web 3.0 environment incorporates Genetic algorithms with support vector machines and user feedbacks for image retrieval purposes.

These approaches for CBIR problem mentioned-above are not applicable to solve the problem of geo-tagged image search on road network because they just only consider the visual content during the retrieval, ignoring the geographical information during the query processing.

\subsection{Spatial Keyword Queries}
There are a great number of studies on spatial keywords queries techniques area, which is a hotspots interested by spatial database community.

\textbf{Sanpshot query}. Spatial keyword snapshot query~\cite{DBLP:conf/icde/ZhangZZL13} technique can be divided into three categories: text priority index, spatial priority index and loosely structured index . The text first index query is to extract relevant inverted files using text index, and then use spatial index to do spatial filtering. For a given spatial keyword query, we can perform an incremental nearest neighbor search on the extracted inverted document~\cite{DBLP:journals/tods/HjaltasonS99}, Until the $k$ objects that satisfies the keywords are found. Text first index query techniques include IF-R$^*$~\cite{DBLP:conf/cikm/ZhouXWGM05}, S2I~\cite{DBLP:conf/ssd/RochaGJN11}, I$^3$~\cite{DBLP:conf/edbt/ZhangTT13}, SFC-QUAD~\cite{DBLP:conf/cikm/ChristoforakiHDMS11}, IL-Quadtree~\cite{DBLP:conf/icde/ZhangZZL13,DBLP:journals/tkde/ZhangZZL16}.

Spatial priority index is to use spatial index to prune space, and then extract corresponding inverted files from keywords. It includes R-tree with inverted files, R-tree with bitmap files, and grid with inverted files. R-tree with inverted files technique comprises R$^*$-IF~\cite{DBLP:conf/cikm/ZhouXWGM05}, KR$^*$-tree~\cite{DBLP:conf/ssdbm/HariharanHLM07}, IR-tree and its variations~\cite{DBLP:conf/ssdbm/HariharanHLM07,DBLP:journals/pvldb/CongJW09,DBLP:journals/vldb/WuCJ12} like IR-tree(Li)~\cite{DBLP:journals/tkde/LiLZLLW11}, WIR-tree~\cite{DBLP:journals/tkde/WuYCJ12}, LBAK-tree~\cite{DBLP:conf/gis/AlsubaieeBL10}. R-tree with bitmap files index includes IR2-tree~\cite{DBLP:conf/icde/FelipeHR08}, SKI~\cite{DBLP:conf/ssdbm/CaryWR10}, bR-tree~\cite{DBLP:conf/icde/ZhangCMTK09}, and MHR-tree~\cite{DBLP:conf/icde/YaoLHH10}. IRGI~\cite{DBLP:conf/icde/ChenC2013} is the only approach for grid with inverted files, which is proposed to address the spatial keywords query problem of wireless data broadcast environment.

Loosely structured index~\cite{DBLP:conf/icde/ZhangOT10,DBLP:conf/sigir/ZhangCT14} constructs spatial index and textual index respectively from spatial dimension and textual dimension. However, There is no connection or loose connection between the spatial index and the textual index. During the query processing, it search out objects that satisfy spatial or time constraints from spatial index and textual index respectively, And then return the results by taking their intersection. Zhang et al.~\cite{DBLP:conf/icde/ZhangOT10} utilized R-tree to index spatial information, but each node adds a tag that records the path from the root node to this node. Their method applies inverted file to index documents, but the objects in the inverted file also carry information that can be used for spatial judgement. Zhang et al.~\cite{DBLP:conf/sigir/ZhangCT14} only used an inverted table to organize spatial information and textual information based on CA algorithm~\cite{DBLP:journals/jcss/FaginLN03}. They proposed a Top-$k$ query algorithm to solve Top-$k$ aggregate queries based on spatial keywords.

\textbf{Continuous query}.
Wu et al.~\cite{DBLP:conf/icde/WuYJC11} first time proposed spatial keyword continuous $k$ nearest neighbor queries. They used the doubly weighted Voronoi unit as the security area of the query to ensure that the user's activities are limited to this security area. For the problem that need to use special sort function in the solution introduced by Wu et al, Huang et al.~\cite{DBLP:conf/cikm/HuangLTF12} proposed a widely used ranking function based approach that works better in terms of efficiency and communication costs. For the problem that the security area query processing method is difficult to achieve the optimal update frequency and single update cost simultaneously, and existing spatial keyword continuous nearest neighbor query processing methods are based on sequential computation model design, Li et al.~\cite{DBLP:journals/pvldb/Li0QYZ014} presented a spatial keyword continuous $k$ neighbor query processing method based on influence set. However, these studies just focused on the spatial textual queries, rather than image retrieval with geo-tagged information.

\subsection{Spatial Keyword Queries on Road Network}
Spatial keyword queries on road network is an other hot issue in the area of spatial data, which can reflect the environment of daily life more realistically. Although the research on keyword query techniques on road network started late, some important achievements have been proposed. Like spatial keyword queries, this query problem can be divided into two categories: snapshot query and continuous query.

\textbf{Sanpshot query on road network}. Rocha-Junior et al.~\cite{DBLP:conf/edbt/Rocha-JuniorN12} studied keyword query on road network space in the first time. They proposed spatial textual index composed of space components, adjacency components, mapping components and inverted file components. And they proposed to build overlay network on the road network to improve query performance. For the same problem, Fang et al.~\cite{DBLP:conf/adc/FangZSWXLC15} proposed a hybrid index SG-Tree which combines G-tree with signature file. Gao et al.~\cite{DBLP:journals/tkde/GaoQZ015} studied a spatial keyword reverse kNN query on the road network and developed an algorithm based on filter refining framework. Furthermore, they introduced a count tree to improve query efficiency. Luo et al.~\cite{DBLP:journals/kbs/LuoJLWLL16} and Lin et al.~\cite{DBLP:conf/icde/LinXH16} studied  spatial keyword reverse KNN query on the road network. Luo et al.~\cite{DBLP:journals/kbs/LuoJLWLL16} proposed an algorithm based on network extension, and an algorithm to make use of the characteristics of network Voronoi graph. Lin et al.~\cite{DBLP:conf/icde/LinXH16} designed a hybrid index structure KcR-tree to store and summarize the space and keyword information of objects. And then they proposed three kinds of optimization techniques for query. In our previous work~\cite{DBLP:conf/edbt/ZhangZZLCW14}, we studied the keyword diversity query on road network and developed inverted index based on signature files. Besides, a segmentation-based method is used to improve the validity of signature files, and an efficient incremental and diversified spatial keyword search algorithm is designed.

\textbf{Continuous query on road network}.
For the continuous query problem, Li et al.~\cite{DBLP:journals/Huazhong/Li2013} studied the problem of spatial keyword continuous top-$k$ query in road network. They proposed a data structure consisting of a PMR-quad tree and three memory tables to store and retrieve relevant information of the road network and objects. Moreover, an adjustable formula for calculating the comprehensive distance is proposed by them to satisfy the different emphasis on keyword similarity and road network distance in various practical applications. The query result is corrected by monitoring the change of the comprehensive distance value of the candidate, so as to realize the continuous processing of the query. Guo et al.~\cite{DBLP:journals/geoinformatica/GuoSAT15} proposed the notion of safety road segment, i.e., when the query is moving on the road, the top-$k$ query results remain unchanged. They presented two algorithms for continuous query in road network QCA and OCA.

Apparently, the researches above just consider the textual similarity and road network distance proximity, They do not consider the situation that the user need to search a image with geo-tag on the road network. To the best of our knowledge, we are the first to study the problem of continuous top-$k$ geo-visual query on road network.

\section{Preliminary}
\label{preliminary}

In this section, we firstly review two basic approaches of image representation, namely invariant feature transform and bag-of-visual-words. Then we formally define the problem of interactive search for geo-tagged image and some concepts. Furthermore, we introduce the outline of our framework. Table~\ref{tab:notation} summarizes the notations frequently used throughout this paper to facilitate the discussion.

\begin{table}
	\centering
    \small
	\begin{tabular}{|p{0.18\columnwidth}| p{0.73\columnwidth} |}
		\hline
		\textbf{Notation} & \textbf{Definition} \\ \hline\hline
		~$\mathcal{D}_I$                                 & A given database of geo-tagged images                \\ \hline
        ~$I_k$                                           & The $k$-th geo-tagged images                 \\ \hline
        ~$\mathcal{O}$                                   & A given dataset of geo-visual object                \\ \hline
        ~$o_k$                                           & The $k$-th geo-visual object in $\mathcal{O}$        \\ \hline
	  	~$\mathcal{G}(\mathcal{N},\mathcal{E},\mathcal{W})$  & The graph model of road network         \\ \hline
        ~$n_i$                                           & The $i$-th node in a road network          \\ \hline
        ~$e_{i,j}$                                       & A edge connecting node $n_i$ and $n_j$        \\ \hline
        ~$w(e_{i,j})$                                    & The weight of edge $e_{i,j}$                  \\ \hline
        ~$o_k.\lambda$                                   & The geo-location information descriptor of $o_k$     \\ \hline
		~$o_k.\psi$                                      & A visual content descriptor of $o_k$            \\ \hline
        ~$p$                                             & A location point on road network             \\ \hline
        ~$\mathcal{V}$                                   & A vector of visual words    \\ \hline
        ~$v$                                             & A visual word        \\ \hline
        ~$N$                                             & A network node        \\ \hline
        ~$\mathcal{I}$                                   & A VIG-Tree index        \\ \hline
        ~$\mathcal{Q}$                                   & A top-$k$ geo-visual query on road network             \\ \hline
        ~$\mathcal{P}(n_i,n_j)$                          & A path connecting $n_i$ and $n_j$             \\ \hline
        ~$\mathcal{P}_s(n_i,n_j)$                        & The shortest path connecting $n_i$ and $n_j$             \\ \hline
        ~$\delta(n_i,n_j)$                               & The travel distance between $n_i$ and $n_j$             \\ \hline
        ~$\delta(\mathcal{P}(n_i,n_j))$                  & The distance of path $\mathcal{P}(n_i,n_j)$          \\ \hline
        ~$\mu$                                           & A parameter used to balance the importance between road network distance proximity and visual content similarity. \\ \hline
        ~$X$                                             & The longitude of a geo-location              \\ \hline
        ~$Y$                                             & The latitude of a geo-location              \\ \hline
        ~$k$                                             & The number of final results              \\ \hline
        ~$\mathcal{F}(\mathcal{Q},o)$                    & Score function measuring the relevance of $\mathcal{Q}$ and $o$ \\ \hline
        ~$Dia(\mathcal{G}(\mathcal{N},\mathcal{E},\mathcal{W}))$ & The diameter of road network $\mathcal{G}(\mathcal{N},\mathcal{E},\mathcal{W})$      \\ \hline
        ~$Dst(\mathcal{Q}.\lambda,o.\lambda)$            & The road network distance proximity of $\mathcal{Q}$ and $o$                               \\ \hline	
        ~$\mathcal{R}$                                   & The result set of T$k$GVQ on road network      \\ \hline
        ~$Sim(\mathcal{Q}.\psi,o.\psi)$                  & The visual content similarity between $\mathcal{Q}$ and $o$.                     \\ \hline
        ~$\mathcal{L}(e_{i,j})$                          & The safe segment of edge $e_{i,j}$                \\ \hline
        ~$\mathcal{S}$                                   & A safe interval      \\ \hline
	\end{tabular}
    \caption{The summary of notations} \label{tab:notation}	
\end{table}

\subsection{Problem Definition}
\begin{definition}[\textbf{Road Network}] \label{def:igis roadnetwork}
Without loss of generality, A road network is modeled as a simple weighted undirected planar graph represented by $\mathcal{G}(\mathcal{N},\mathcal{E},\mathcal{W})$ and $\mathcal{N},\mathcal{E},\mathcal{W} \neq \emptyset$, wherein $\mathcal{N} = \{n_1,n_2,...,n_{|\mathcal{N}|}\}$ is the set of nodes which represent road intersections or road endpoints, and $|\mathcal{N}|$ is the total number of nodes. The set of edges $\mathcal{E} = \{e_{i,j}|i,j \in \textbf{N}^+, \mbox{and } i \in [1,|\mathcal{N}|], j \in [1,|\mathcal{N}|], i \neq j\}$ denotes all the road segments connecting two nodes, i.e., $e_{\alpha,\beta}=(n_\alpha,n_\beta)$ represents the road segment connecting node $n_\alpha$ and $n_\beta$. Let $p$ be a location point on $\mathcal{G}(\mathcal{N},\mathcal{E},\mathcal{W})$, if $p$ is located on edge $e_{i,j}$, we denote it as $p \in e_{i,j}$. $\mathcal{W} = \{w(e_{i,j})|e_{i,j} \in \mathcal{E}\}$ denotes the set of weights which represent the travel distance of road segments or trip time, i.e., $w(e_{i,j})$ is the trip distance or time cost traveling through from node $n_i$ to $n_j$. To simplifying the presentation, we use travel distance hereinafter. We define the travel distance between any two positions $p$ and $p'$ on road networks as $\delta(p,p')$, and for an edge $e_{i,j}$, $w(e_{i,j}) = \delta(n_i,n_j)$. Furthermore, in this paper we assume that each edge $e_{i,j} \in \mathcal{E}$ is bi-directional, and the distance or time cost of it is irrelevant to the direction. It is described formally in Assumption~\ref{asmp:direction}.
\end{definition}

\newtheorem{assumption}{Assumption}
\begin{assumption}\label{asmp:direction}
Given a road network $\mathcal{G}(\mathcal{N},\mathcal{E},\mathcal{W})$ and $\mathcal{N},\mathcal{E},\mathcal{W} \neq \emptyset$. $\forall e_{i,j} \in \mathcal{E}$, $\exists e_{i,j} \equiv e_{j,i}$ and $w(e_{i,j}) \equiv w(e_{j,i})$.
\end{assumption}

All the conceptions and algorithms which will be stated in the following sections are meaningful and effective based on Assumption~\ref{asmp:direction} is this work.

\begin{definition}[\textbf{Shortest Path}] \label{def:igis shortestpath}
Given a road network $\mathcal{G}(\mathcal{N},\mathcal{E},\mathcal{W})$, $n_i, n_j \in \mathcal{N}$ are two node. A path between $n_i$ and $n_j$ is defined as $\mathcal{P}(n_i,n_j) = \{e^{(1)},e^{(2)},...,e^{(K)}\}$, $K$ is the number of edges in this path. Its distance is $\delta(\mathcal{P}(n_i,n_j))$. Obviously, the distance of a path is equal to the sum of weights of all edges in this path, i.e.,
\begin{equation}\label{equ:distanceofpath}
\delta(\mathcal{P}(n_i,n_j)) = \sum\limits_{k=1}^{K}w(e^{(k)})
\end{equation}
where $K$ is the number of edges in the path $\mathcal{P}(n_i,n_j)$, $e^{(k)}$ represents the $k$-th edge. Based on the notion of path, we define the shortest path connecting $n_i$ and $n_j$ as $\mathcal{P}_s(n_i,n_j)$, which has the smallest distance among all paths connecting these two nodes. Formally, the distance of $\mathcal{P}_s(n_i,n_j)$ is described as:
\begin{equation}\label{equ:shortestpath}
\delta(\mathcal{P}_s(n_i,n_j)) = min(\{\sum\limits_{k=1}^{K}w(e^{(k)})|K \in \textbf{N}^+, \mbox{and } K \in [1,|\mathcal{E}|], e^{(k)} \in \mathcal{P}(n_i,n_j)\})
\end{equation}
where the function $min(\mathcal{X})$ is to return the minimum value of element in the set $\mathcal{X}$, and $\textbf{N}^+$ represents the set of positive integer.
\end{definition}

\begin{figure}[thb]
\newskip\subfigtoppskip \subfigtopskip = -0.1cm
\begin{minipage}[b]{1\linewidth}
\begin{center}
     \subfigure[{The model of road network and shortest path}]{
     \includegraphics[width=0.48\linewidth]{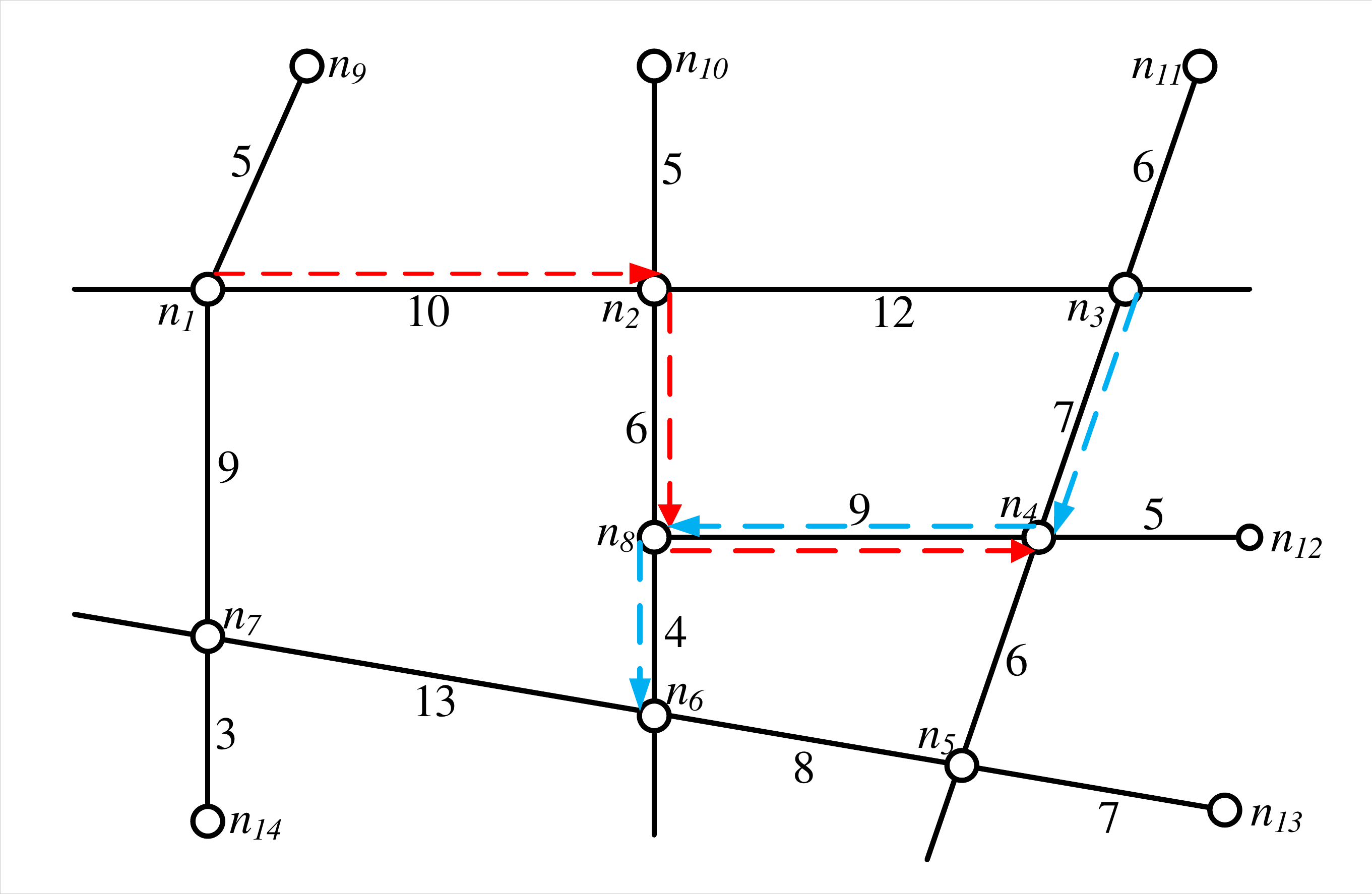}
     }
     \subfigure[{A continuous top-$k$ geo-visual query}]{
     \includegraphics[width=0.48\linewidth]{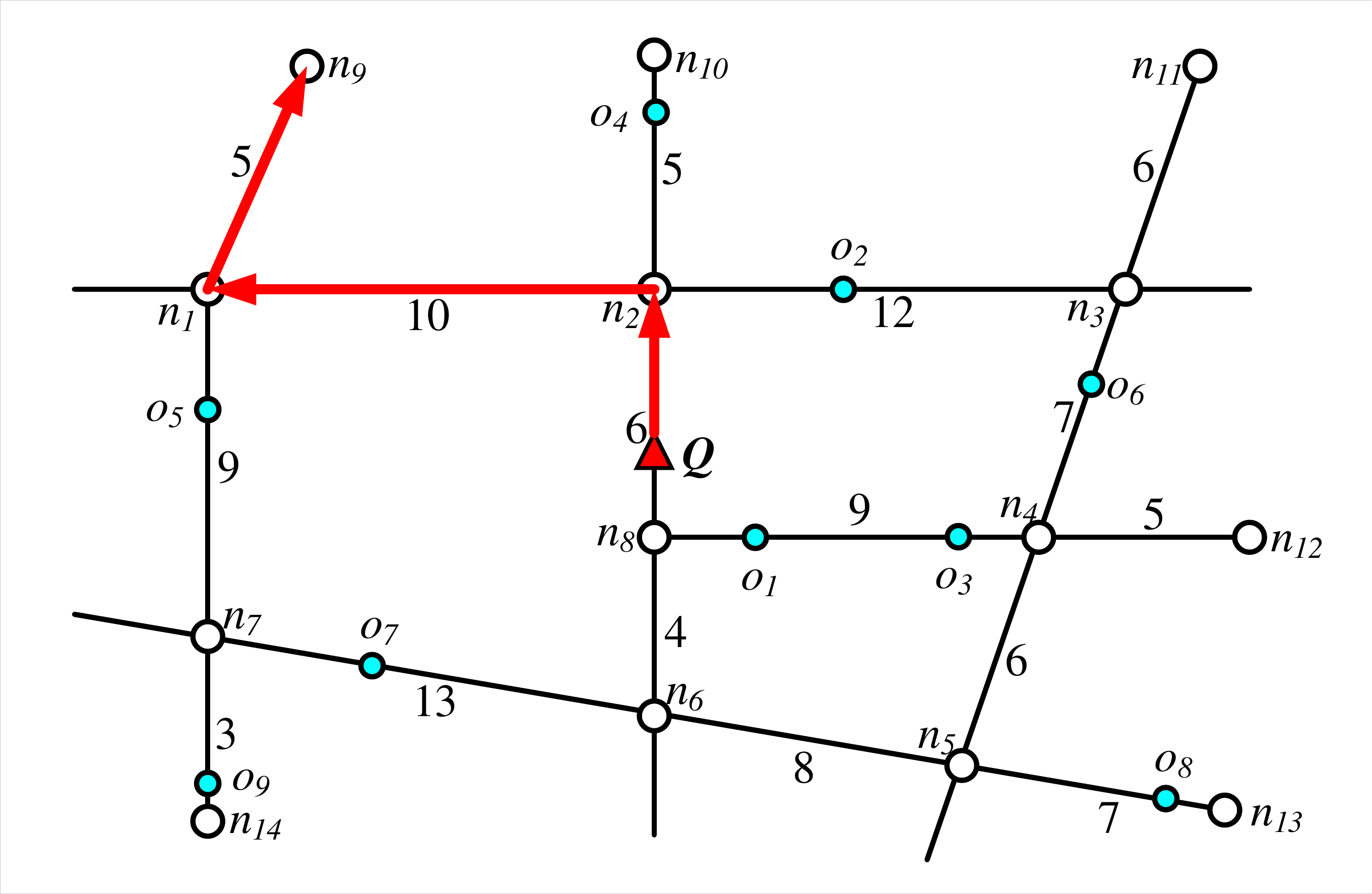}
     }
     \captionsetup{justification=centering}
\vspace{-1mm}
\caption{\small An example of road network and continuous top-$k$ geo-visual query}
\label{fig:fig3}
\end{center}
\end{minipage}
\end{figure}

\newtheorem{exmp}{Example}
\begin{exmp}\label{exmp:roadnetwork}
Figure.~\ref{fig:fig3}(a) shows an example of road network contains 14 nodes and 16 edges. The set of node is $\mathcal{N} = \{n_1,n_2,n_3,n_4,n_5,n_6,n_7,n_8,n_9,n_{10},n_{11},n_{12},n_{13},n_{14}\}$. The weight of each edge is shown aside, e.g., $w(e_{1,2})=10$, $w(e_{6,8})=4$. The shortest path between $n_1$ and $n_4$ consists of $e_{1,2}$, $e_{2,8}$ and $e_{8,4}$, which is highlight by red dashed line with arrow. The distance of it is $\delta(\mathcal{P}_s(n_1,n_4)) = 25$. The blue dashed line denotes the shortest path connecting $n_3$ and $n_6$, $\delta(\mathcal{P}_s(n_3,n_6)) = 20$.
\end{exmp}

\begin{definition}[\textbf{Geo-visual Object}] \label{def:igis geovisual-obj}
Given a geo-tagged image database $\mathcal{D}_I$ which storages $|\mathcal{O}|$ image, i.e., $\mathcal{D}_I = \{I_1,I_2,...,I_{|\mathcal{D}_I|}\}$. A geo-visual objects dataset is defined as $\mathcal{O}=\{o_1, o_2,..., o_{|\mathcal{O}|}\}$ containing $|\mathcal{O}|$ objects. Each geo-visual object $o_k \in \mathcal{O}$ resides an edge $e_{i,j}$ of a road network, and it is associated with a geographical information descriptor $o_k.\lambda$ and a visual content descriptor $o_k.\psi$. More specifically, $o_k.\lambda$ is represented by a 2-dimensional geo-location extracted from the geo-tag of image $I_k$, i.e., $o_k.\lambda = (X,Y)$, wherein $X$ and $Y$ are longitude and latitude respectively. The visual content descriptor is defined as $o_k.\psi = (v_1,v_2,...,v_m)$, which is a visual words vector generated from a geo-tagged image $I_k$ by low-level feature extraction.

If a road network contains several geo-visual objects which are located on edges, we can denote it as $\mathcal{G}(\mathcal{N},\mathcal{E},\mathcal{W},\mathcal{O})$. As the definition presented above, $\mathcal{O}$ is the set of geo-visual objects. On the other hand, we define the distance between a geo-visual object $o_k$ and endpoint nodes $n$ and $n'$ of an edge on which it lies as $\delta(\mathcal{P}(o_k,n))$ and $\delta(\mathcal{P}(o_k,n'))$ based on Definition~\ref{def:igis shortestpath}. The shortest distance between two objects $o_k$ and $o_l$ is denoted as $\delta(\mathcal{P}_s(o_k,o_l))$.
\end{definition}

\begin{definition}[\textbf{Top-$k$ Geo-visual Query (T$k$GVQ)}] \label{def:igis topkgvq}
Given a road network with geo-visual objects $\mathcal{G}(\mathcal{N},\mathcal{E},\mathcal{W},\mathcal{O})$, a top-$k$ geo-visual query is defined as $\mathcal{Q}=(\lambda,\psi,k)$, in which $\mathcal{Q}.\lambda$ denotes the location of query, $\mathcal{Q}.\psi$ is a visual works vector generated from the query image, and $\mathcal{Q}.k$ is the number of requested objects. \textbf{T$k$GVQ} aims to search out $k$ geo-visual objects $\mathcal{R}=\{o_1,o_2,...,o_k\}$ from $\mathcal{G}(\mathcal{N},\mathcal{E},\mathcal{W},\mathcal{O})$, which are ranked according to the score measured by function $\mathcal{F}(\mathcal{Q},o)$, i.e.,
\begin{equation*}\label{equ:TkGVQ}
\mathcal{R}=\{\mathcal{R} \subseteq \mathcal{O},|R|=k|\forall o \in \mathcal{O},o' \in \mathcal{O} \setminus \mathcal{R}, \mathcal{F}(\mathcal{Q},o)>\mathcal{F}(\mathcal{Q},o')\}
\end{equation*}
and the score function is defined as follows:
\begin{equation}\label{equ:scorefunc}
\mathcal{F}(\mathcal{Q},o) = \mu Dst(\mathcal{Q}.\lambda,o.\lambda)+(1-\mu)Sim(\mathcal{Q}.\psi,o.\psi)
\end{equation}
where function $Dst(\mathcal{Q}.\lambda,o.\lambda)$ measures the distance proximity between $\mathcal{Q}$ and $o$, function $Sim(\mathcal{Q}.\psi,o.\psi)$ calculates the visual content relevance between these two visual words vectors, and $\mu \in [0,1]$ is used to balance the importance between road network distance proximity and visual content similarity. If $\mu > 0.5 $, the road network distance proximity is more important than visual content similarity, and conversely, If $\mu < 0.5$, the visual content similarity plays a more considerable role. Note that the geo-visual objects with the \textbf{small score values} are preferred (i,e., ranked higher).
\end{definition}

According to the definition of T$k$GVQ, the score of geo-visual objects is determined by road network distance proximity and visual content similarity, which are measured by function $Dst(\mathcal{Q}.\lambda,o.\lambda)$ and $Sim(\mathcal{Q}.\psi,o.\psi)$ respectively. next we describe these two functions in formal to explain how to calculate the distance proximity and visual content similarity. Firstly, an important notion named road network diameter is presented in the following definition.

\begin{definition}[\textbf{Road Network Diameter}] \label{def:igis rnd}
Given a road network $\mathcal{G}(\mathcal{N},\mathcal{E},\mathcal{W})$, the road network diameter $Dia(\mathcal{G}(\mathcal{N},\mathcal{E},\mathcal{W}))$ is defined as follows:
\begin{equation}\label{equ:rndiameter}
Dia(\mathcal{G}(\mathcal{N},\mathcal{E},\mathcal{W})) = max(\{\delta(\mathcal{P}_s(n_i,n_j))|\forall n_i,n_j \in \mathcal{N}\})
\end{equation}
where the function $max(\mathcal{X})$ is to return the maximum value of element in the set $\mathcal{X}$. It is apparent that the diameter of a road network is the maximum shortest path connecting any two node.
\end{definition}

Based on the conception of road network diameter and shortest path, next we propose the definition of road network distance proximity.

\begin{definition}[\textbf{Road Network Distance Proximity}] \label{def:igis rndp}
Given a road network with geo-visual objects $\mathcal{G}(\mathcal{N},\mathcal{E}, \mathcal{W},\mathcal{O})$, a geo-visual object $o \in \mathcal{O}$ and a top-$k$ geo-visual query $\mathcal{Q}$, the road network distance proximity measure is defined as:
\begin{equation}\label{equ:rndistance}
Dst(\mathcal{Q}.\lambda,o.\lambda) = \frac{\delta(\mathcal{P}_s(\mathcal{Q}.\lambda,o.\lambda))}{Dia(\mathcal{G}(\mathcal{N},\mathcal{E}, \mathcal{W},\mathcal{O}))}
\end{equation}
where $\delta(\mathcal{P}_s(\mathcal{Q}.\lambda,o.\lambda))$ is the distance of shortest path connecting $\mathcal{Q}$ and $o$. Obviously, the road network distance proximity $Dst(\mathcal{Q}.\lambda,o.\lambda)$ is determined by the shortest path between the query location and the object. Besides, it is easily to prove that the value range of $Dst(\mathcal{Q}.\lambda,o.\lambda)$ is $[0,1]$.
\end{definition}

\begin{proof}
We prove the proposition $Dst(\mathcal{Q}.\lambda,o.\lambda) \in [0,1]$ by contradiction. (1)Firstly, we assume that $Dst(\mathcal{Q}.\lambda,o.\lambda) > 1$, i.e., $\delta(\mathcal{P}_s(\mathcal{Q}.\lambda,o.\lambda)) > $ $Dia(\mathcal{G}(\mathcal{N},\mathcal{E}, \mathcal{W},\mathcal{O}))$. Thus, according to Definition~\ref{def:igis rnd}, $\delta
(\mathcal{P}_s(\mathcal{Q}.\lambda,o.\lambda)) > max(\{\delta(\mathcal{P}_s(n_i,n_j))|\forall n_i,n_j \in \mathcal{N}\})$. Without loss of generality, assume $\mathcal{Q}$ is located on the edge $e_{\alpha,\beta} = (n_\alpha,n_\beta) \in \mathcal{E}$ and $o$ is located on the edge $e_{\gamma,\iota} = (n_\gamma,n_\iota) \in \mathcal{E}$, and $\mathcal{P}_s(n_\alpha,n_\iota)$ contains $e_{\alpha,\beta}$ and $e_{\gamma,\iota}$. Thus $\delta(\mathcal{P}_s(\mathcal{Q}.\lambda,o.\lambda)) \leq \delta(\mathcal{P}_s(n_\alpha,n_\iota))$. It is clearly that $\mathcal{P}_s(n_\alpha,n_\iota) \in \{\delta(\mathcal{P}_s(n_i,n_j))|\forall n_i,n_j \in \mathcal{N}\}$, $\delta(\mathcal{P}_s(n_\alpha,n_\iota)) \leq$ $max(\{\delta(\mathcal{P}_s(n_i,n_j))|\forall n_i,n_j \in \mathcal{N}\})$. Thus, $\delta(\mathcal{P}_s(\mathcal{Q}.\lambda,o.\lambda)) \leq max(\{\delta(\mathcal{P}_s(n_i,n_j))|\forall n_i,n_j \in \mathcal{N}\})$. This is contradictory to the assumption that $\delta(\mathcal{P}_s(\mathcal{Q}.\lambda,o.\lambda)) > Dia(\mathcal{G}(\mathcal{N},\mathcal{E}, \mathcal{W},\mathcal{O}))$. (2)It is easily to know that the distance of shortest path and network diameter are positive, and if $\mathcal{Q}$ and $o$ at the same place, then $\delta(\mathcal{P}_s(\mathcal{Q}.\lambda,o.\lambda)) = 0$, i.e., $\delta(\mathcal{P}_s(\mathcal{Q}.\lambda,o.\lambda)) \geq 0$ and $Dia(\mathcal{G}(\mathcal{N},\mathcal{E}, \mathcal{W},\mathcal{O})) > 0$. Therefore $Dst(\mathcal{Q}.\lambda,o.\lambda) \geq 0$. $\hfill\blacksquare$
\end{proof}

The other important part of score function, visual content similarity function $Sim(\mathcal{Q}.\psi,o.\psi)$, relies on Bag-of-Visual-Words (BoVWs for short) model and SIFT technique to. The local features of an image which are extracted from image key-points are encoded into a visual words vector.

SIFT is used to detect and describe local features of an image, which transforms an image into a large collection of local feature vectors. These local features are invariant to image translation, scaling, rotation, and partially invariant to illumination changes and affine or 3-dimension projection~\cite{DBLP:journals/ijcv/Lowe04}. It has four main stages introduced as follows:

\textbf{(1)Scale-space extrema detection}. The first stage is named scale-space extrema detection, which searches all scales and image location. The algorithm can recognize the potential interest points by using a difference-of-Gaussian (DoG) function convolved with the image.

\textbf{(2)Keypoint localization}. The second stage is keypoint localization aiming to precisely localize the keypoints. In this stage, two kinds of extremum point will be discarded: low contrast extreme points and unstable edge response points.

\textbf{(3)Orientation assignment}. The third stage is named orientation assignment. In this stage, one or more orientations are designated to each keypoint based on gradient directions of the local image. all computations are performed in a scale-invariant manner for each image. The dominant orientation of the keypoint is designated by the highest peak in the histogram after orientation histogram formed by computing the magnitude and orientation within a region around the keypoint.

\textbf{(4)Keypoint descriptor}. The last stage, named keypoint descriptor, which measures the gradients of local image at the selected scale in the region around each keypoint. These gradients are transformed into a representation, which allows for relatively large local shape distortion and illumination changes.

An image is represented as a vector of visual words denoted as $\mathcal{V}$ constructed by vector quantization of feature descriptors utilizing BoVWs extending from bag-of-words (BoW for short) technique which is widely used in textual similarity measurement~\cite{DBLP:conf/iccv/SivicZ03}. In this approach, $k$-means is applied to create the visual words. In other words, visual words are defined by the $k$-means cluster centers and the SIFT features in every images are then assigned to the nearest cluster center to give a visual word representation~\cite{DBLP:conf/bmvc/ChumPZ08}.

Based on the notion introduced above, the similarity of two images $I_1$ and $I_2$ can be calculated by the following equation:
\begin{equation}\label{equ:similarityof2img}
Sim(\mathcal{V}_1,\mathcal{V}_2) = (1-\frac{|\mathcal{V}_1 \cap \mathcal{V}_2|}{|\mathcal{V}_1 \cup \mathcal{V}_2|})
\end{equation}
where $\mathcal{V}_1$ and $\mathcal{V}_2$ are two sets of visual words generated from $I_1$ and $I_2$. In this method, all the words are equally important. Obviously,

According to the conception of BoVWs and the image similarity measurement, we propose the definition of visual content similarity to measure the visual relevance between two geo-visual objects.

\begin{definition}[\textbf{Visual Content Similarity}] \label{def:igis vcs}
Given a T$k$GVQ denoted as $\mathcal{Q}$ and a geo-visual object $o$, the visual content similarity between $\mathcal{Q}$ and $o$ is measured by the following function:
\begin{equation}\label{equ:vcsimilarity}
Sim(\mathcal{Q}.\psi,o.\psi) = (1-\frac{|\mathcal{Q}.\psi \cap o.\psi|}{|\mathcal{Q}.\psi \cup o.\psi|})
\end{equation}
\end{definition}

\begin{definition}[\textbf{Continuous Top-$k$ Geo-visual Query (CT$k$GVQ)}] \label{def:igis ctopkgvq}
Given a road network with geo-visual objects $\mathcal{G}(\mathcal{N},\mathcal{E},\mathcal{W},\mathcal{O})$, a continuous top-$k$ geo-visual query is defined as $\mathcal{Q}_c = (\lambda,\psi,k,[\tau_s,\tau_e])$, wherein $\lambda$,$\psi$ and $k$ have the same meaning with the relevant part in Definition~\ref{def:igis topkgvq}, $[\tau_s,\tau_e]$ is a time range in which a CT$k$GVQ is running. For each new location $\mathcal{Q}.\lambda$, CT$k$GVQ returns a new result set $\mathcal{R}$.
\end{definition}

\begin{exmp}
To demonstrate the CT$k$GVQ problem plainly, we consider a simple example in Figure.~\ref{fig:fig3}(b). A CT$k$GVQ denoted as $\mathcal{Q}$ is represented by a red triangle which is moving along the route marked by red line with arrow. Clearly, it will pass through node $n_2$, $n_1$ and $n_9$ successively. During the movement, the result of $\mathcal{Q}$ will be updated continuously. When $\mathcal{Q}$ is near $n_2$, the result set is $\{o_2,o_4\}$. When it moves to $n_1$, the result set will be $\{o_5\}$.
\end{exmp}

\section{Hybrid Indexing for Network Aware Continuous Image Retrieval}
\label{hybridIndex}

In this section, we present a \textbf{V}isual \textbf{I}nverted \textbf{G}-tree (\vig for short) that supports the following required functions for continuous geo-visual search and ranking on road network: 1)\textbf{visual filtering}: all the visually irrelevant nodes and objects have to be discarded as early as possible to cut down the search cost; 2)\textbf{network filtering}: all the nodes,which are farther in networks, have to access as later as possible to avoid unnecessary network expansion. 3)\textbf{relevance computation and ranking}: since only the top-k geo-visual objects are returned and $k$ is expected to be much smaller than the total number of match objects, it is desirable to have an incremental search process that integrates the computation of the joint relevance, and object ranking seamlessly so that the search process can stop as soon as the top-$k$ objects are identified.

\subsection{Hybrid Index Framework: VIG-Tree}
\label{vigtree}

\vig is a combination of an G-Tree and visual inverted index. In particular, each node of an \vig contains both
network information and visual words information; the former consists of a distance matrix and a set of minimum bounding areas (MBR) of children nodes, and the latter in the form of a visual inverted index file for the edges or nodes rooted at the node.

Fig.~\ref{fig:fig4} presents the basic indexing architecture of \vig. In the \vig, a leaf node L consists of four parts: MBR component, matrix component, visual component and a subnetwork component. In the following, we describe each component in more detail.

\begin{figure}[thb]
\newskip\subfigtoppskip \subfigtopskip = -0.1cm
\begin{minipage}[b]{1\linewidth}
\begin{center}
     \subfigure[{Graph partition}]{
     \includegraphics[width=0.9\linewidth]{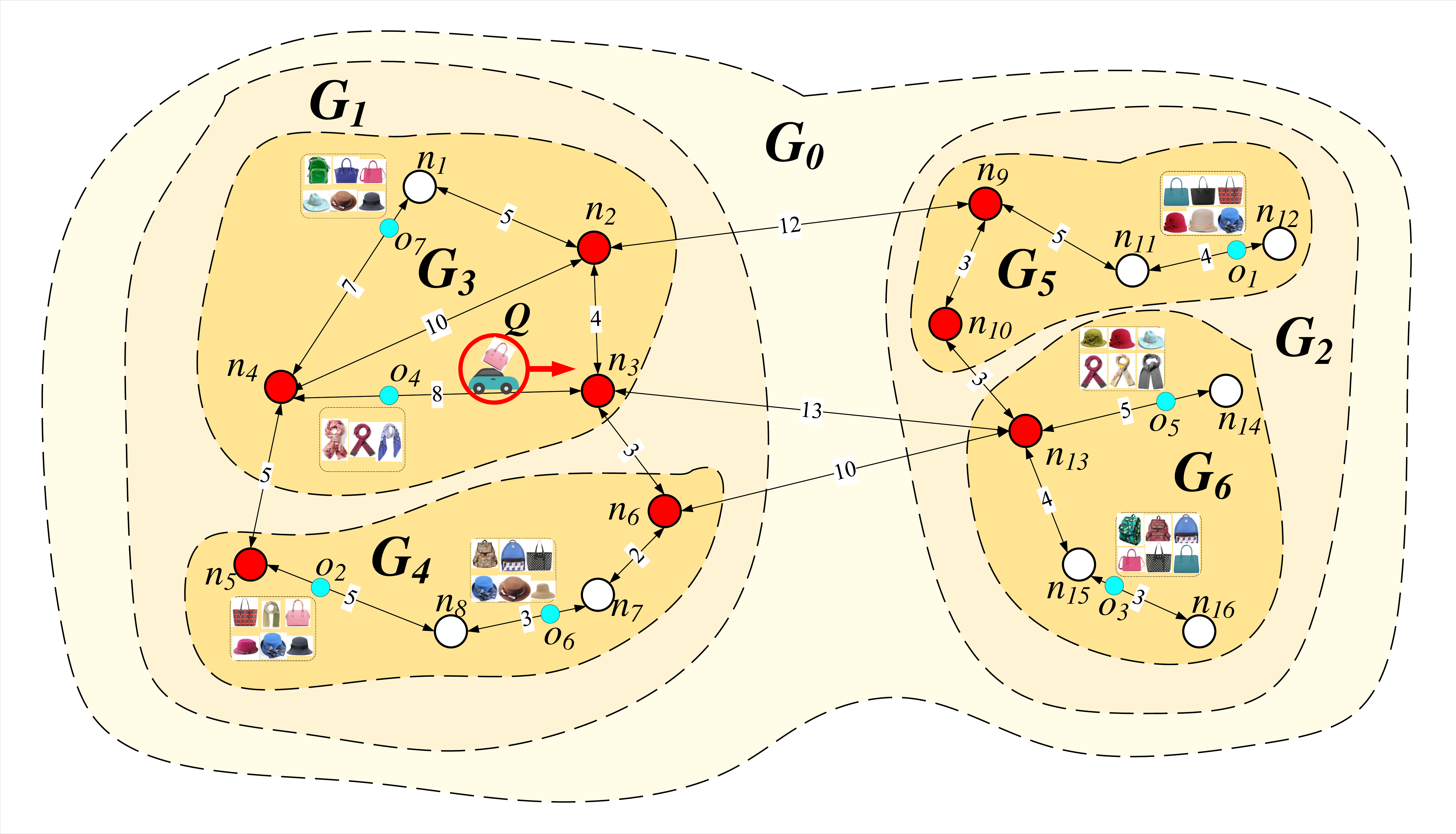}
     }
     \subfigure[VIG-Tree]{
     \includegraphics[width=0.9\linewidth]{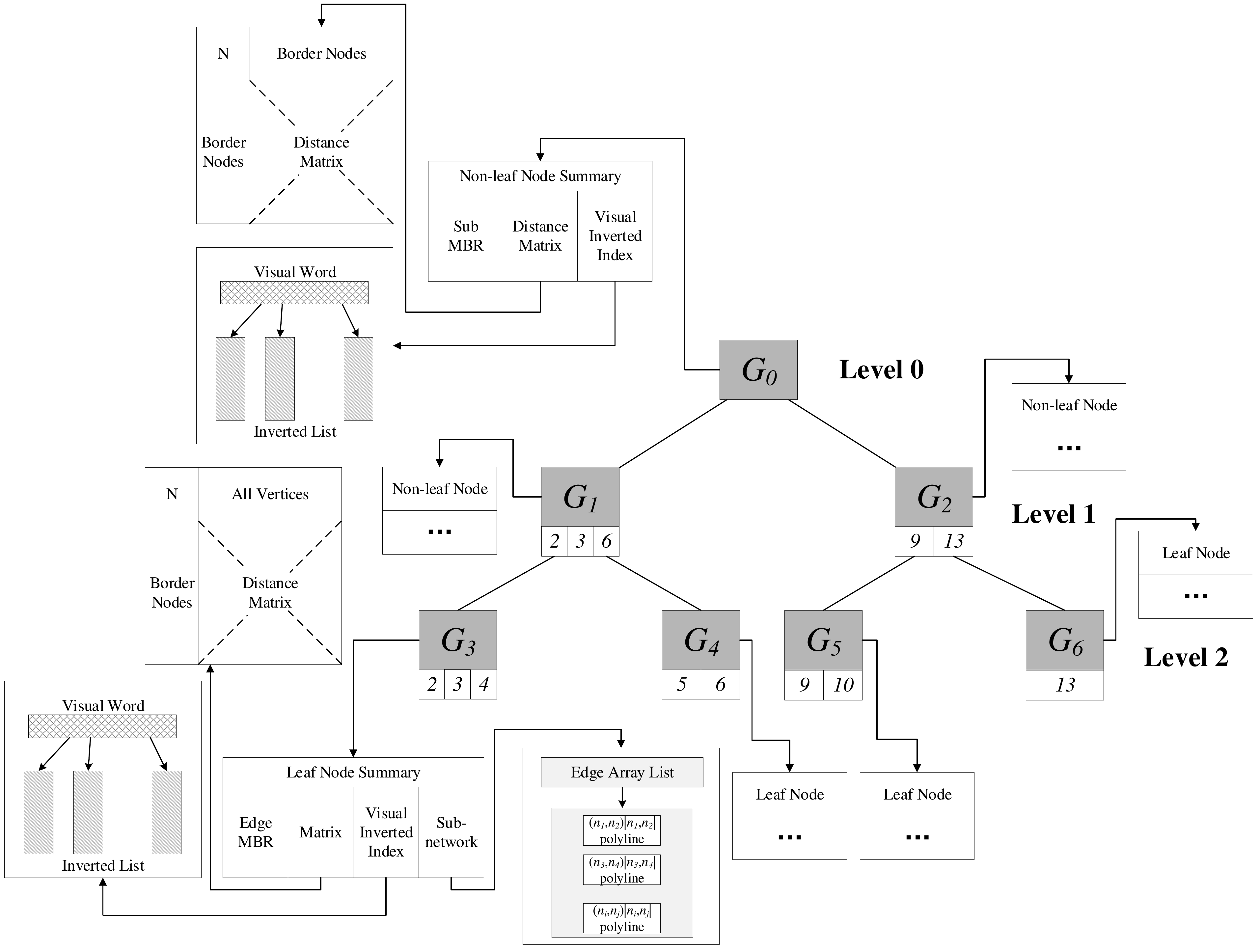}
     }
     \captionsetup{justification=centering}
\vspace{-1mm}
\caption{\small A VIG-Tree}
\label{fig:fig4}
\end{center}
\end{minipage}
\end{figure}

\begin{itemize}
\item{\textbf{MBR component.}} The MBR component contains a set of MBRs that encloses the corresponding edges which roots at current node, and are used to find the edges that cover the original query location.

\item{\textbf{Matrix component.}} The matrix component consists of three parts, the rows are all borders in the node, the columns are all vertices that are rooted at this node, the entry of distant matrix records the shortest-path distance between the border and the vertex.

\item{\textbf{Visual component.}} The visual component consists of a list of inverted indexes of unique visual words of leaf node. Each visual inverted index corresponds to a visual word $v$ and pointing to a list of edges that contain $v$.

\item{\textbf{Subnetwork component.}} The subnetwork component combines network connectivity information and detail visual word information of geo-visual objects of each edge. It points to an edge array list of current node. The edge array consists of the edge id (e.g., ($n_1$, $n_2$)), the length of the edge (e.g., ($\mid n_1$, $n_2 \mid$)), and a detail polyline to describe the edge. The detail polyline is used to locate the edge where the query lies, and start a network expansion from the nodes of the edge.
\end{itemize}

A non-leaf node N of \vig is composed of MBR component, Matrix component, and visual component. To be specific, the MBR component contains a set of MBRs that encloses the corresponding nodes which roots at current node. Similar to leaf node, the matrix component of non-leaf node also consists of three parts, but the major difference is that both rows and columns are  all borders that rooted at this node, and the value of each entry of distance matrix is the shortest-path distance between the two borders. Similarly, The visual component of non-leaf node is composed of a set of visual inverted indexes which is generated by aggregating the visual information from its children node, and points to a list of nodes that contain corresponding visual word.

Similar to ~\cite{DBLP:conf/edbt/Rocha-JuniorN12}, we also used a B-tree to manage the edge information. The key of the B-tree is composed of edge id and visual word id to the inverted list that contains the geo-visual objects lying
on the edge with visual word $v$ in their description. This key value points to an inverted list which stores the geo-visual objects lying on the corresponding edge that have a visual word $v$ in their description.For each geo-visual objects, the object id and the network distance between the geo-visual objects and the related node of the edge are stored. This information is used to compute the visual content similarity between the geo-visual object and query.

\subsection{Processing of Geo-visual Queries on Road Network}\label{sec:snapshot query}
We proceed to present an important metric, the minimum visual network distance $minVND$, which will be used in the geo-visual query processing. Given a geo-visual query $q$ and a network node $N$ in the \vig, the metric $minVND$ offers a lower bound on the actual visual network distance between query $\mathcal{Q}$ and the geo-visual objects enclosed in the rectangle of network node $N$. The search space of \vig can be efficiently pruned by this bound.

\begin{definition} [$minVND(\mathcal{Q},N)$]
The distance of a geo-visual query $\mathcal{Q}$ from a node $N$ in
the \vig, denoted as $minVND(\mathcal{Q},N)$, is defined as follows:
\begin{equation}\label{eq:vig mindst}
\begin{aligned}
minVND(\mathcal{Q},N)=\mu minND(\mathcal{Q}.\lambda,N.b.\lambda)+ \\
(1-\mu)minVD(\mathcal{Q}.\psi,N.\psi)
\end{aligned}
\end{equation}
where
$minND(\mathcal{Q}.\lambda, N.b.\lambda)$ is the length of the shortest path between $\mathcal{Q}.\lambda$ and $N.b.\lambda$,
$minVD(\mathcal{Q}.\psi,N.\psi)$ is the minimum visual content relevance between $\mathcal{Q}.\psi$ and $N.\psi$.
\end{definition}

\begin{theorem}\label{lemma:vig point_prune}
Given a geo-visual query $\mathcal{Q}$, a node $N$, and a set of geo-visual objects $\mathcal{O}$ in node $N$, for any $o\in\mathcal{O}$, we have $\mathcal{F}(\mathcal{Q},o)\leq \mathcal{F}(\mathcal{Q},N)$.
\end{theorem}
\begin{proof}
Since  geo-visual object $o$ is enclosed by the border nodes of node $N$, the
minimum network distance between $\mathcal{Q}.\lambda$ and $ N.b.\lambda$ is no larger than the minimum network distance between $\mathcal{Q}.\lambda$ and $o.\lambda$:

\begin{equation*}
minND(\mathcal{Q}.\lambda, N.b.\lambda) \leq minND(\mathcal{Q}.\lambda, o.\lambda)
\end{equation*}
As the visual word list of the node is generated by aggregating all the unique visual word of its geo-visual objects. Hence:
\begin{equation*}
minVD(\mathcal{Q}.\psi,N.\psi) \leq minVD(\mathcal{Q}.\psi, o.\psi)
\end{equation*}
According to Equation \ref{equ:scorefunc}, Equation \ref{equ:rndistance} and Equation \ref{equ:vcsimilarity}, we obtain:
\begin{equation*}
\mathcal{F}(\mathcal{Q},o)\leq \mathcal{F}(\mathcal{Q},N)
\end{equation*}
thus completing the proof.
\end{proof}

When searching the \vig for the $k$ geo-visual objects nearest to a
query $\mathcal{Q}$, the first step is to decide which node should enter first. Metric $minVND$ offers an approximation
of the visual network distance ranking score to every children in the node and, therefore, can
be used to guide the search.

To process geo-visual queries with \vig, we
exploit the best-first traversal algorithm for retrieving
the top-k geo-visual objects. With the best-first traversal algorithm, a priority
queue is used to keep track of the nodes and objects that have yet to
be visited. When deciding which node to visit next, the algorithm picks the
node $N$ on road network with the smallest $minVND(\mathcal{Q},N)$ value in the set of all
nodes that have yet to be visited. The algorithm terminates when $k$
nearest objects (ranked according to Equation \ref{equ:scorefunc}) have been found.

\begin{algorithm}
\begin{algorithmic}[1]
\footnotesize
\caption{\bf Geo-visual Search on Road Network($\mathcal{Q}$, $k$, $\mathcal{I}$) }
\label{alg:skt search}
\INPUT $\mathcal{Q}~:$ the geo-visual query, $k~:$ the number of geo-visual image return,
   $\mathcal{I}~:$ \vig index
\OUTPUT $\mathcal{R}:$ top-$k$ query result results

\STATE $\mathcal{R}:= \emptyset$; $\mathcal{H} = \emptyset$, $\gamma_{max} = \infty$

\STATE $\mathcal{H}\leftarrow$ new a min first heap

\label{alg:skt_build fsig}
\STATE $\mathcal{H}$.Enqueue($\mathcal{I}.root, minVND(\mathcal{Q},\mathcal{I}.root)$)
\label{alg:skt push root}

\WHILE{ $\mathcal{H} \not = \emptyset$ and $\mid \mathcal{R}\mid \leq k$}
     \label{alg:skt_loop_s}
     \STATE $n \leftarrow$ the node popped from $\mathcal{H}$
     \label{alg:skt_leaf s}
     \IF{$n$ is a leaf node}
        \label{alg:skt_leaf ss}
        \FOR{ each edge $e$ in node $n$}
            \IF{$minVND(\mathcal{Q}, e) \leq \gamma_{max}$}
                \STATE $\mathcal{H}$.Enqueue($e',  minVND(\mathcal{Q}, e)$)
                \label{alg:skt_leaf e}
            \ENDIF
        \ENDFOR
     \ENDIF
     \IF{$n$ is an edge}
        \label{alg:skt_edge s}
        \FOR{ each object $o$ in edge $e$ and $o.\psi \cap \mathcal{Q}.\psi \not = \emptyset$}
            \IF{$minVND(\mathcal{Q}, o) \leq \gamma_{max}$}
                \STATE update $\mathcal{R}$ by $(o, minVND(\mathcal{Q}, o) )$
                             \label{alg:skt_edge e}
            \ENDIF
        \ENDFOR
     \ENDIF
     \IF{$n$ is an non-leaf node}
        \FOR{each child $n'$ in node $n$}
            \label{alg:skt_non-leaf s}
            \IF{$minVND(\mathcal{Q}, n') \leq \gamma_{max}$}
               \STATE $\mathcal{H}$.Enqueue($n',  minVND(\mathcal{Q}, n')$)
               \label{alg:skt_non-leaf e}
            \ENDIF
        \ENDFOR
     \ENDIF
\ENDWHILE
\label{alg:skt_loop_e}
\RETURN{$\mathcal{R}$}
\label{alg:skt_return}
\end{algorithmic}
\end{algorithm}

Algorithm~\ref{alg:skt search} illustrates the details of the \vig based geo-visual query on road network. A minimum heap $\mathcal{H}$ is employed to keep the \vig's nodes where the key of a node is its minimum visual network ranking score. For the input query, we find out the root node of current time segment, calculate the minimal spatial temporal visual ranking score for the root node, and then pushed the root node into the $\mathcal{H}$ in Line ~\ref{alg:skt push root}. Then the purposed algorithm executes the while loop (Line~\ref{alg:skt_loop_s}-\ref{alg:skt_loop_e})until the top-$k$ results are ultimately reported in Line~\ref{alg:skt_return}.


In each iteration, the top node $n$ with minimum visual network ranking score is popped from $\mathcal{H}$. When the popped node is a leaf node(Line~\ref{alg:skt_leaf s}), the enclosed edge $e$ of $n$ will be pushed to $\mathcal{H}$ if its minimal visual network distance ranking score between $e$ and $\mathcal{Q}$, denoted by $minVND(\mathcal{Q}, e)$, is not larger than $\gamma_{max}$ (Line~\ref{alg:skt_leaf ss}-~\ref{alg:skt_leaf e}). When the popped object is an edge(Line~\ref{alg:skt_edge s}), for the geo-visual objects which have same visual word with $\mathcal{Q}$ and belong to current edge, if its minimal visual network distance ranking score between geo-visual $o$ and $\mathcal{Q}$ is not larger than $\gamma_{max}$, we push $o$ into result set $\mathcal{R}$ and add update $\gamma_{max}$. Otherwise, if the popped object is a non-leaf node(Line~\ref{alg:skt_non-leaf s}), a child node $n^{'}$ of $n$ will be pushed to $\mathcal{H}$ if its minimal visual network distance ranking score between $n^{'}$ and $\mathcal{Q}$, denoted by $minVND(\mathcal{Q}, n')$, is not larger than $\lambda_{max}$ (Line~\ref{alg:skt_non-leaf s}-~\ref{alg:skt_non-leaf e}). The algorithm terminates when $\mathcal{H}$ is empty and the results are kept in $\mathcal{R}$.

\section{Moving Monitor Algorithm}
\label{mmalgorithm}

In this section, we propose a efficient algorithm to solve the problem of continuous top-$k$ geo-visual objects query on road network. At first we introduce the relevant notion named safe interval which plays a key role in this solution. A result updating rule is designed to guide the result set updating during movement of query. After that we describe the moving monitor algorithm in detail.

\subsection{Safe Interval}
If applying the solution of snapshot query in the continuous top-$k$ geo-visual object query on road network, the cost of computation is quite high due to the continuous movement of $\mathcal{Q}$. In order to solve this problem effectively, in this paper we propose a novel conception on road network named safe interval which is the base of our solution. Firstly we introduce a lemma which is the base of the notion of safe interval, and then propose the definition of safe interval.

\begin{lemma}\label{lemma:ssegr}
Given a CT$k$GVOQ $Q$ located on the edge $e_{i,j}$, the result set of $Q$ is denoted as $\mathcal{R}(Q)$ which consists of three part: (1)the set of relevant objects located on the edge $e_{i,j}$ denoted as $\mathcal{O}_{i,j}$, (2)the result set of T$k$GVOQ on the location $n_i.\lambda$ denoted as $\mathcal{R}(\mathcal{Q}.\lambda = n_i.\lambda)$ and (3) the result set of T$k$GVOQ on the location $n_j.\lambda$ denoted as $\mathcal{R}(\mathcal{Q}.\lambda = n_j.\lambda)$, that is,
\begin{equation*}
\mathcal{R}(Q) = \mathcal{O}_{i,j} \cup \mathcal{R}(\mathcal{Q}.\lambda = n_i.\lambda) \cup \mathcal{R}(\mathcal{Q}.\lambda = n_j.\lambda)
\end{equation*}
where relevant object refers to object that matches at least one of the query visual word.
\end{lemma}

\begin{proof}
We prove Lemma~\ref{lemma:ssegr} by contradiction. We assume that $ \exists o \in \mathcal{R}(\mathcal{Q})$ and $o \notin (\mathcal{O}_{i,j} \cup \mathcal{R}(\mathcal{Q}.\lambda = n_i.\lambda) \cup \mathcal{R}(\mathcal{Q}.\lambda = n_j.\lambda))$. The shortest path connecting $\mathcal{Q}$ and $o$ denoted as $\mathcal{P}_s(\mathcal{Q}.\lambda,o.\lambda)$ must pass through one of $n_i$ and $n_j$ since $o \notin \mathcal{O}_{i,j}$. In general, we assume that $o$ is closer to $n_i$. $\forall o' \in \mathcal{R}(\mathcal{Q}.\lambda = n_i.\lambda)$. As $o \notin \mathcal{R}(\mathcal{Q}.\lambda = n_i.\lambda)$, it is easy to know that $\mathcal{F}(\mathcal{Q}.\lambda = n_i.\lambda,o)>\mathcal{F}(\mathcal{Q}.\lambda = n_i.\lambda,o')$. Since $\delta(\mathcal{P}(\mathcal{Q},o))$ can be added to the road network distance proximity component of both sides of the above inequality and the visual content similarity component remains
constant, we have $\mathcal{F}(\mathcal{Q}, o)>\mathcal{F}(\mathcal{Q},o')$. It means that $\forall o' \in \mathcal{R}(\mathcal{Q}.\lambda=n_i.\lambda)$, $\mathcal{F}(\mathcal{Q}.\lambda=n_i.\lambda,o')>\mathcal{F}(\mathcal{Q},o)$. Therefore, $o$ should not be in $\mathcal{R}(\mathcal{Q})$. This is contradictory to the initial assumption that $o \in \mathcal{R}(\mathcal{Q})$.$\hfill\blacksquare$
\end{proof}

\begin{definition}[\textbf{Safe Segment}] \label{def:igis safeseg}
Given a road network $\mathcal{G}(\mathcal{N},\mathcal{E},\mathcal{W})$ and a CT$k$GVQ denoted as $\mathcal{Q}$ which is moving on an edge $e_{i,j} \in \mathcal{E}$. A safe segment of $e_{i,j}$ is defined as $\mathcal{L}(e_{i,j})$, which is a subsegment of $e_{i,j}$ and if $\mathcal{Q}$ is located on it, the results is not change. Formally, $\forall p_\alpha,p_\beta $ $\in$ $\mathcal{L}(e_{i,j})$, $\mathcal{R}(\mathcal{Q}(\lambda=p_\alpha)) \equiv \mathcal{R}(\mathcal{Q}(\lambda=p_\beta))$, where $p_\alpha$ and $p_\beta$ are two location and  $\mathcal{R}(\mathcal{Q}(\lambda=p_\alpha))$ denotes the result set of query $\mathcal{Q}$ at the location $p_\alpha$.
\end{definition}

Based on safe segment, a client-server architecture is utilized in this paper. When a client submits a continuous top-$k$ geo-visual object query $\mathcal{Q}$, the edge on which this query $\mathcal{Q}$ locates can be searched out and the result sets of query at the two end nodes of this edge are returned at first. After that the top-$k$ results of $\mathcal{Q}$ can be generated from the result sets of the end nodes. It is only when the client leaves out from the safe segment, It is need to send a new location to the server and then repeats the above process.

According the definition above-mentioned, we can know that each edge has its safe segment, in which the result set of a query do not need to be updated. This can greatly reduce computing cost. Moreover, we find a situation in a road network that if a long path containing serval segment, and it has no crossroad, we can consider these safe segments as a whole one in which the result set still do not need to be updated. There is no doubt that it can further cut down the cost when we utilize this conception in  CT$k$GVQ. We call it as safe interval and give the definition as follows.

\begin{definition}[\textbf{Safe Interval}] \label{def:igis safeintv}
Given a road network $\mathcal{G}(\mathcal{N},\mathcal{E},\mathcal{W})$, we define safe interval as follows: assume that $n_i,n_j \in \mathcal{N}$ and the path connecting them is denoted as $\mathcal{P}(n_1,n_j) = \{e^{(1)},e^{(2)},...,e^{(K)}\}$ which has no branch road, the safe segments of each edge in $\mathcal{P}(n_1,n_j)$ are denoted as $\mathcal{L}(e^{(1)}),\mathcal{L}(e^{(2)}),...,\mathcal{L}(e^{(K)})$, the safe interval of this path is defined as $\mathcal{S} = \{\mathcal{L}(e^{(1)}),\mathcal{L}(e^{(2)}),...,\mathcal{L}(e^{(K)})\}$.
\end{definition}

\begin{theorem}[\textbf{Results Updating Rule}]
Given a CT$k$GVQ denoted as $Q$ located on the edge $e_{i,j}$ moves towards the direction from $n_i$ to $n_j$, the result set will be updated under the following rule: (1)the geo-visual objects $\hat{o} \in \mathcal{R(\mathcal{Q})}$ which are located on the position behind $\mathcal{Q}$ will be discarded. (2)the geo-visual objects $\bar{o} \notin \mathcal{R}(\mathcal{Q})$ which are located on the position in the front of $\mathcal{Q}$ will be added in. Formally, assume that $\mathcal{Q}$ pass through $p_\alpha$ and $p_\beta$ successively, $p_\alpha,p_\beta \in e_{i,j}$.

(1)If $\delta(\hat{o}.\lambda,p_\alpha) < \delta(\hat{o}.\lambda,p_\beta)$, then
\begin{equation*}
\mathcal{R}(\mathcal{Q}(\lambda=p_\beta)) = \mathcal{R}(\mathcal{Q}(\lambda=p_\alpha)) \setminus \{\hat{o}\}.
\end{equation*}

(2)If $\delta(\bar{o}.\lambda,p_\alpha) > \delta(\bar{o}.\lambda,p_\beta)$, then
\begin{equation*}
\mathcal{R}(\mathcal{Q}(\lambda=p_\beta)) = \mathcal{R}(\mathcal{Q}(\lambda=p_\alpha)) \cup \{\hat{o}\}.
\end{equation*}
\end{theorem}

\begin{exmp}
As shown in Figure.~\ref{fig:fig4}, the query $\mathcal{Q}$ denoted as red triangle located on the edge $e_{1,2}$. Three geo-visual objects denoted by yellow small circle are located on edge $e_{4,1}$, $e_{1,2}$ and $e_{2,3}$ respectively. $\mathcal{Q}$ is moving from location $p_1$ to $p_2$. It is easy to find that $\delta(o_1.\lambda,p_1) < \delta(o_1.\lambda,p_2)$, thus it may will be discarded from the result set. On the other hand, $\delta(o_2.\lambda,p_1)>\delta(o_2.\lambda,p_2)$ and $\delta(o_3.\lambda,p_1)>\delta(o_3.\lambda,p_2)$, that means $o_2$ and $o_3$ may will be added in the result set.
\end{exmp}

\begin{figure}[thb]
\newskip\subfigtoppskip \subfigtopskip = -0.1cm
\centering
\includegraphics[width=0.8\linewidth]{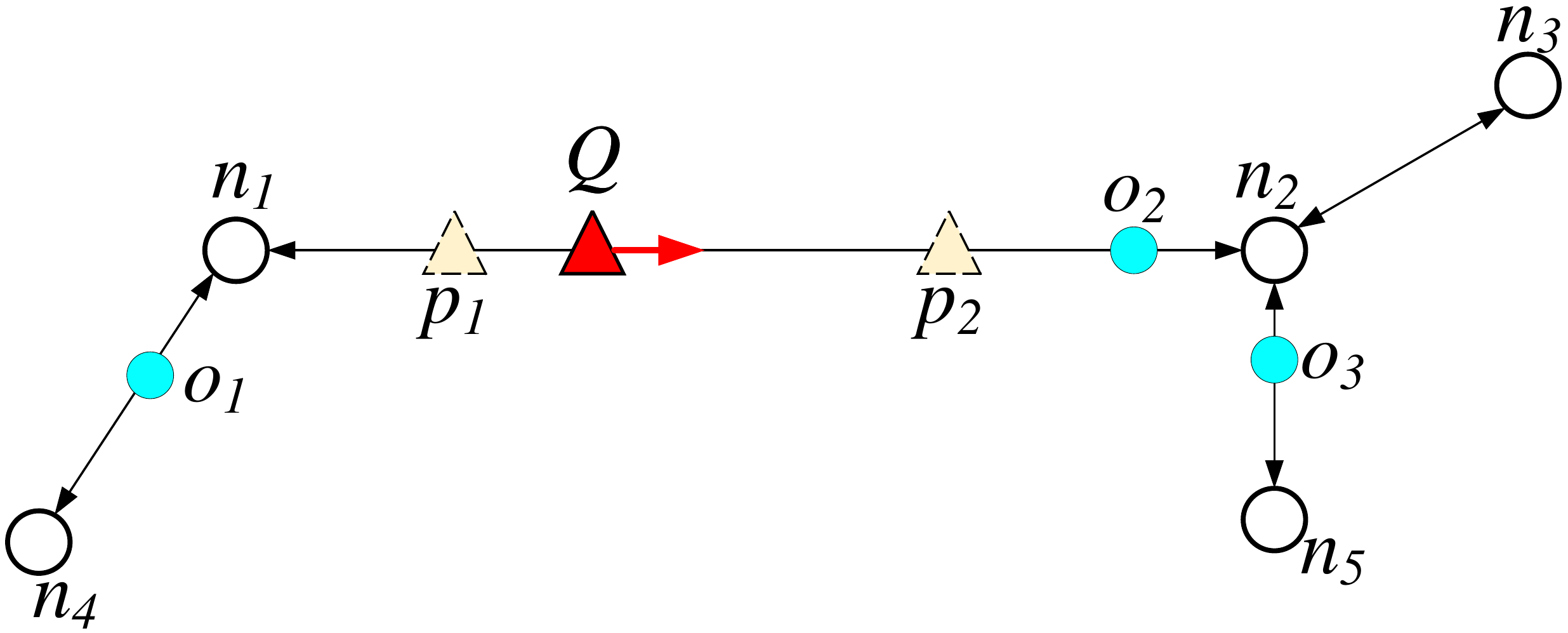}
\vspace{-1mm}
\caption{\small  An example of result updating rule }
\label{fig:fig4}
\end{figure}

\subsection{Moving Monitor Algorithm}
A standard client-server architecture is utilized to imply a continuous top-$k$ geo-visual object query. During the processing, there are two cases must be considered, i.e., (1)\textbf{CASE 1}. The client moves into a safe interval; (2)\textbf{CASE 2}. The client leaves out from a safe interval and moves into a new edge. In order to understand our approach easily, before introducing the query algorithm, we at first present these two case in detail as follows.

\textbf{CASE 1}. When the client moves into a safe interval, first the server returns the candidates set according to Lemma~\ref{lemma:ssegr}. Then in the processing of moving, the client will generate the results set based on candidates set and continuously update it according the result updating rule described above.

\textbf{CASE 2}. When the client leaves out from a safe interval and moves into a new edge, it promptly informs the server to recompute the candidates set for the new edge. After that the client will move into a new safe interval of this edge, and then the results set will be calculated as CASE 1.

Based on the conception of safe segment and the results updating rule, we design an algorithm named \textbf{Moving Monitor Algorithm} (MMA for short) which runs on the client-server architecture. The pseudo-code is shown as follows.

\begin{algorithm}
\begin{algorithmic}[1]
\footnotesize
\caption{\bf Moving Monitor Algorithm (MMA)}
\label{alg:qca}

\INPUT  A road network with geo-visual objects $\mathcal{G}(\mathcal{N},\mathcal{E},\mathcal{W},\mathcal{O})$, a moving query $\mathcal{Q}$.
\OUTPUT A results set $\mathcal{R}$.

\STATE Initializing: $\mathcal{R} \leftarrow \emptyset$; //A results set
\STATE Initializing: $\mathcal{C} \leftarrow \emptyset$; //A candidates set
\STATE Initializing: $\mathcal{S} \leftarrow \emptyset$; //A safe interval
\STATE Initializing: $\zeta \leftarrow \emptyset$; //A variable to save the current location of $\mathcal{Q}$

\FOR{each $\mathcal{Q}.\lambda$}
    \STATE $\mathcal{S} \leftarrow SafeIntervalSearch(\mathcal{Q}.\lambda)$; //Locating the safe interval in which $\mathcal{Q}$ lies on
    \STATE $\mathcal{C} \leftarrow CandidateSearch(\mathcal{S})$;
    \WHILE{$GetLocation(\zeta,\mathcal{Q})$}
        \IF{$\zeta \in \mathcal{S}$}
            \STATE $\mathcal{R} \leftarrow SnapshortTopKQuery(\mathcal{\mathcal{Q}.\lambda},\mathcal{C})$;
            \STATE $SendToClient(\mathcal{R})$; //Send top-$k$ results to client
            \STATE Break;
        \ENDIF
        \IF{$\zeta \notin \mathcal{S}$}
            \STATE $\mathcal{S} \leftarrow SafeIntervalSearch(\mathcal{Q}.\lambda)$; //Locating the safe interval in which $\mathcal{Q}$ is lie on
            \STATE $\mathcal{C} \leftarrow CandidateSearch(\mathcal{S})$;
        \ENDIF
    \ENDWHILE
\ENDFOR

\end{algorithmic}
\end{algorithm}

\section{Experimental Evaluation}
\label{exp}

In this section, we present results of a comprehensive performance study on real
road network datasets to evaluate the efficiency and scalability of the proposed techniques.
Specifically, we evaluate the effectiveness of the following indexing techniques for continuous geo-visual search on road network.
\begin{itemize}
\item{\textbf{Overlay}} Overlay is the road network similarity search technique proposed in ~\cite{DBLP:conf/edbt/Rocha-JuniorN12}.
\item{\textbf{Overlay-SI}} Overlay-SI is a natural extension of the road network similarity search technique in ~\cite{DBLP:conf/edbt/Rocha-JuniorN12} by the safe interval technique, which is proposed in Section~\ref{hybridIndex}.
\item{\textbf{VIG-SI}} VIG-SI is a combination of \vig technique and safe interval technique, which are proposed in Section~\ref{hybridIndex}.
\end{itemize}

\noindent \textbf{Datasets.} Performance of various algorithms is evaluated on both real road network datasets.
The following three datasets are deployed in the experiments.
Road network \textbf{NA} is obtained from the North America Road Network(\url{http://www.cs.utah.edu/~lifeifei/SpatialDataset.htm}) with $175,812$ nodes and $179,178$ road segments; Road network of \textbf{SF} is obtained from San Francisco Road Network (\url{http://www.cs.utah.edu/~lifeifei/SpatialDataset.htm}) where there are $174,955$  nodes and $223,000$ edges; Road network \textbf{AU} is obtained from ~\cite{DBLP:conf/edbt/Rocha-JuniorN12} with $1,181,142$  nodes and $1,631,421$ edges. The locations of the objects are randomly chosen from spatial datasets Rtree-Portal (
\url{http://www.rtreeportal.org}). Note that we move an object to its closest road segment if it does not lie on any edge in the road network.
In the experiments, the locations of all datasets are scaled to the 2-dimensional
space $[0,~10000]^2$. Similarly, we generate image dataset by crawling millions image from photo-sharing site Flickr(\url{http://www.flickr.com/}). The dataset size varies from 400K to 2M to evaluate the scalability of our proposed algorithm. Table~\ref{tab:image} summaries the important statistics of these image datasets.

\begin{table}
\centering
\caption{Information of datasets} \label{tab:image}

\centering
\begin{tabular}{cccc}

\hline
Datasets& Number of Images& Dist. Visual Words Number& Avg. Visual Words Number\\
\hline

400K&	400000&	    613940&  	118.2\\
800K&	800000&	    607401& 	128.6\\
1.2M&	    1200000&	616947& 	125.9\\
1.6M&	    1600000&	614026& 	127.4\\
2M&	2000000&   	618347&  	124.5\\
\hline
\end{tabular}

\end{table}

\vspace{.5mm}
\noindent
{\bf Workload.}
A workload for the continuous geo-visual query consists of $100$ queries. The  query response time, the  number of communication time are employed to evaluate the performance of the algorithms.
The query locations are randomly selected from the locations of the underlying objects, and each query contains a sequence of locations in the form
of ($n_i$, $n_j$, $\lambda$). As both \textbf{Overlay-SI} and \textbf{VIG-SI} adopt segment interval technique, we only take \textbf{Overlay} and \textbf{VIG-SI} into comparison, while comparing the number of communication time. The query length, which indicates the
number of locations the moving object reported, varies from 100 to 500; the number of the query visual words changes from 20 to 100; the number of the returned results grows from 10 to 50; the preference parameter $\mu$ varies from 0.1 to 0.9; the image dataset size grows from 400K to 2M. By default, the query length, the query visual word number, the result number, he preference parameter and dataset size is set to \textbf{100}, \textbf{40}, \textbf{10}, \textbf{0.5} and \textbf{800K} respectively. Experiments are run on a PC with Intel Xeon 2.60GHz dual CPU and 16G memory running Ubuntu. All algorithms in the experiments are implemented in Java.

\noindent \textbf{Evaluation on query length.}
We investigate the query response time of three algorithm use the query length of 100, 200, 300, 400, 500. Figure~\ref{fig:query-length}(a) reports that the performance of three algorithm in terms of the query response
time degrades against
the growth of query length. As expected, \textbf{Overlay-SI} significantly outperforms
\textbf{Overlay} in Figure ~\ref{fig:query-length}(a), since the safe interval can reduce the number of geo-visual search in server. \textbf{VIG-SI} achieves the better performance
compared with \textbf{Overlay-SI} because the G-tree can significantly reduce the network expansion cost. Similarly, in Figure~\ref{fig:query-length}(b), the communication time increasing as the growth of query length, because more query locations mean higher probability to invoke geo-visual search.

\begin{figure*}
\newskip\subfigtoppskip \subfigtopskip = -0.1cm
\begin{minipage}[b]{1\linewidth}
\begin{center}
     \subfigure[Time]{
     \includegraphics[width=0.48\linewidth]{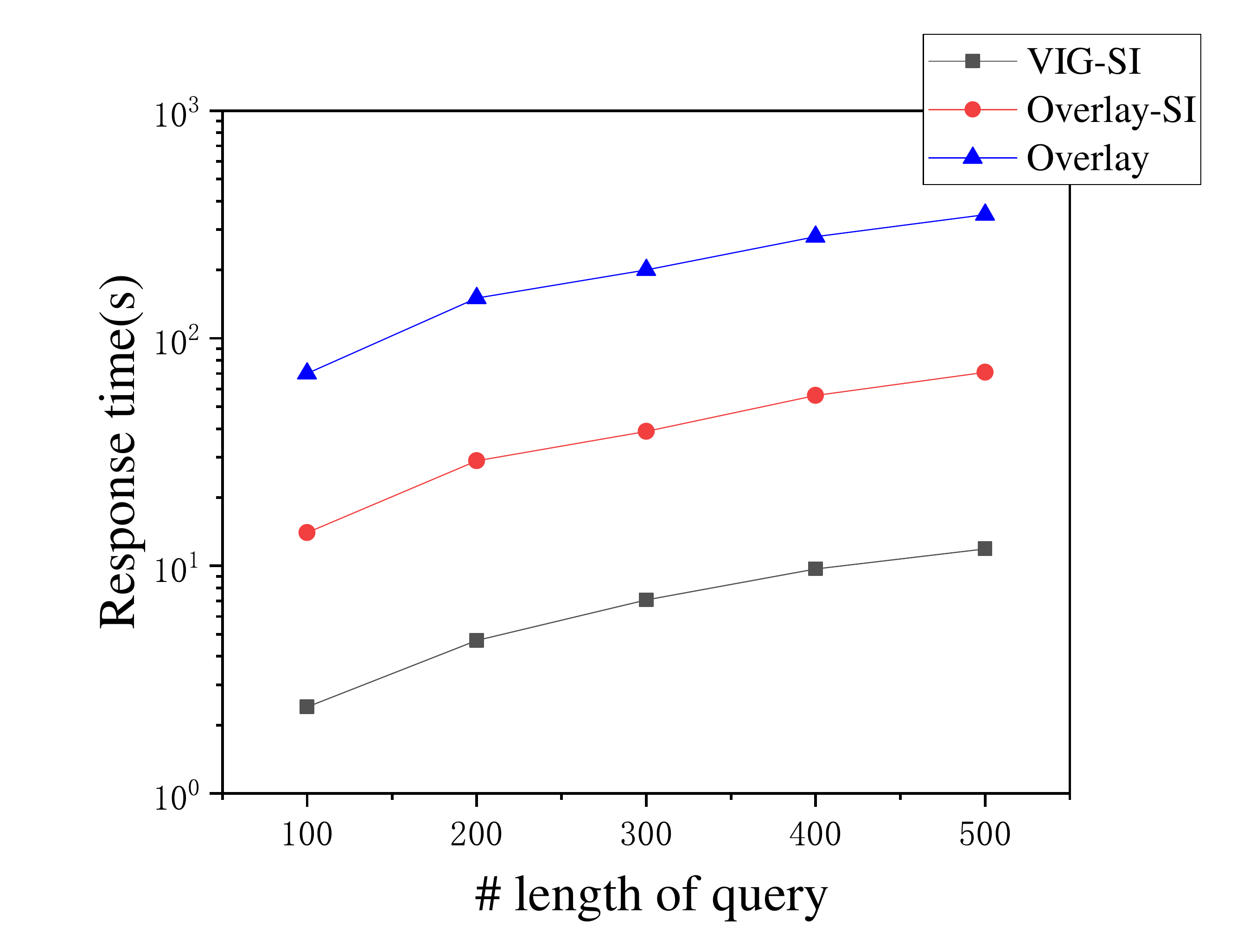}
     }
     \subfigure[Communication Time]{
     \includegraphics[width=0.48\linewidth]{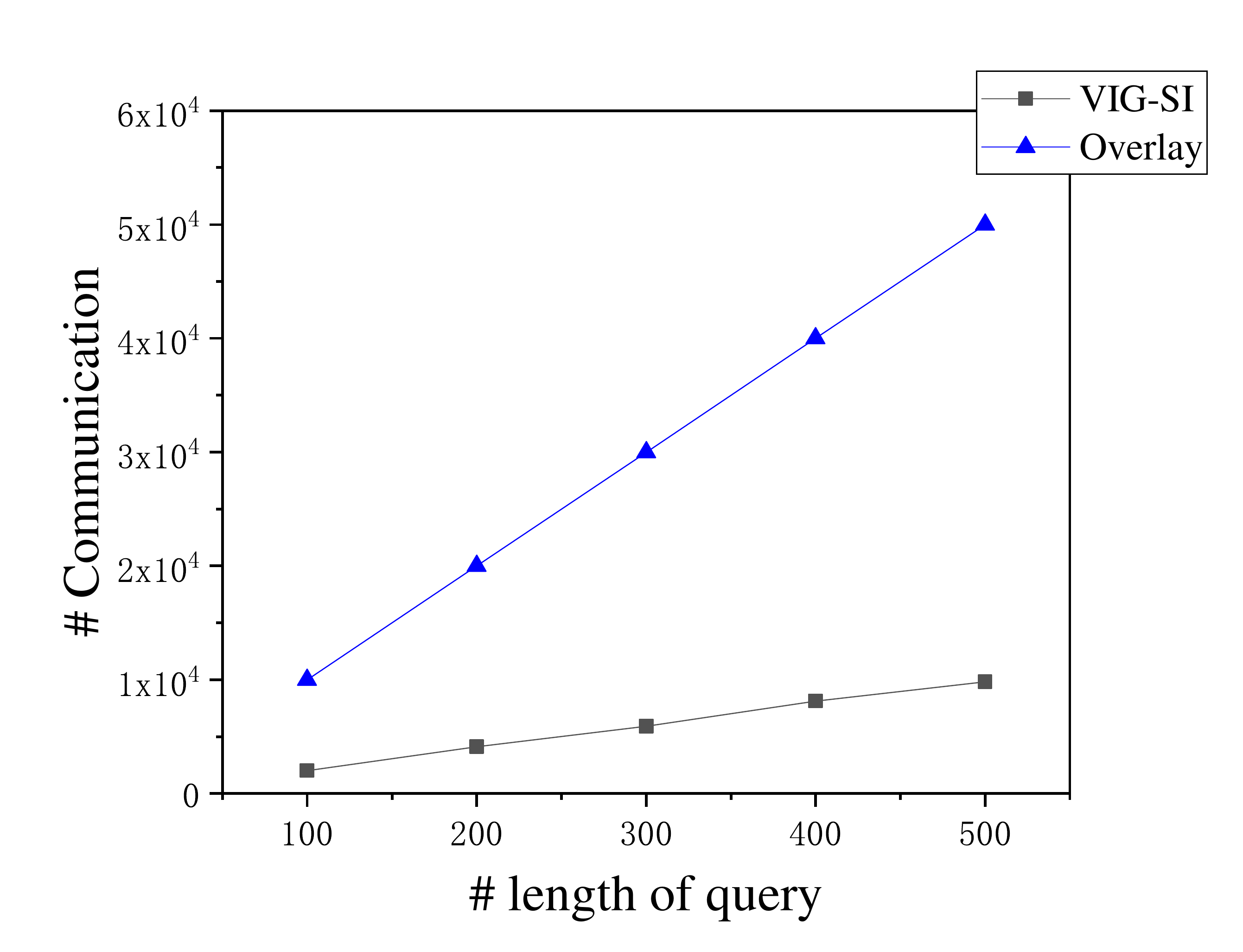}
     }
   \captionsetup{justification=centering}
       \vspace{-0.2cm}
\caption{Effect of query length}
\label{fig:query-length}
\end{center}
\end{minipage}
\end{figure*}

\noindent \textbf{Evaluation on the number of query visual words.}
We evaluate
the effect of the number of query visual words $l$ in Figure ~\ref{fig:query-visual} on
dataset SF, where $l$ varies from 20 to 100. The larger l will lead to the more
chance for each object to meet the keyword constraint, it indicates more nodes or  objects need to take into consideration. Thus, the response time of both three algorithms increases. However, the communication cost of both \textbf{Overlay} and \textbf{VIG-SI} is stable and \textbf{VIG-SI} has less communication. As \textbf{VIG-SI} adopts safe interval technique, the communication time of \textbf{VIG-SI} significantly reduces. Meanwhile, because the \textbf{Overlay} executes geo-visual search in each query location, and the probability to invoke geo-visual search for safe interval will not change as the query visual words increase.

\begin{figure*}
\newskip\subfigtoppskip \subfigtopskip = -0.1cm
\begin{minipage}[b]{1\linewidth}
\begin{center}
     \subfigure[Time]{
     \includegraphics[width=0.48\linewidth]{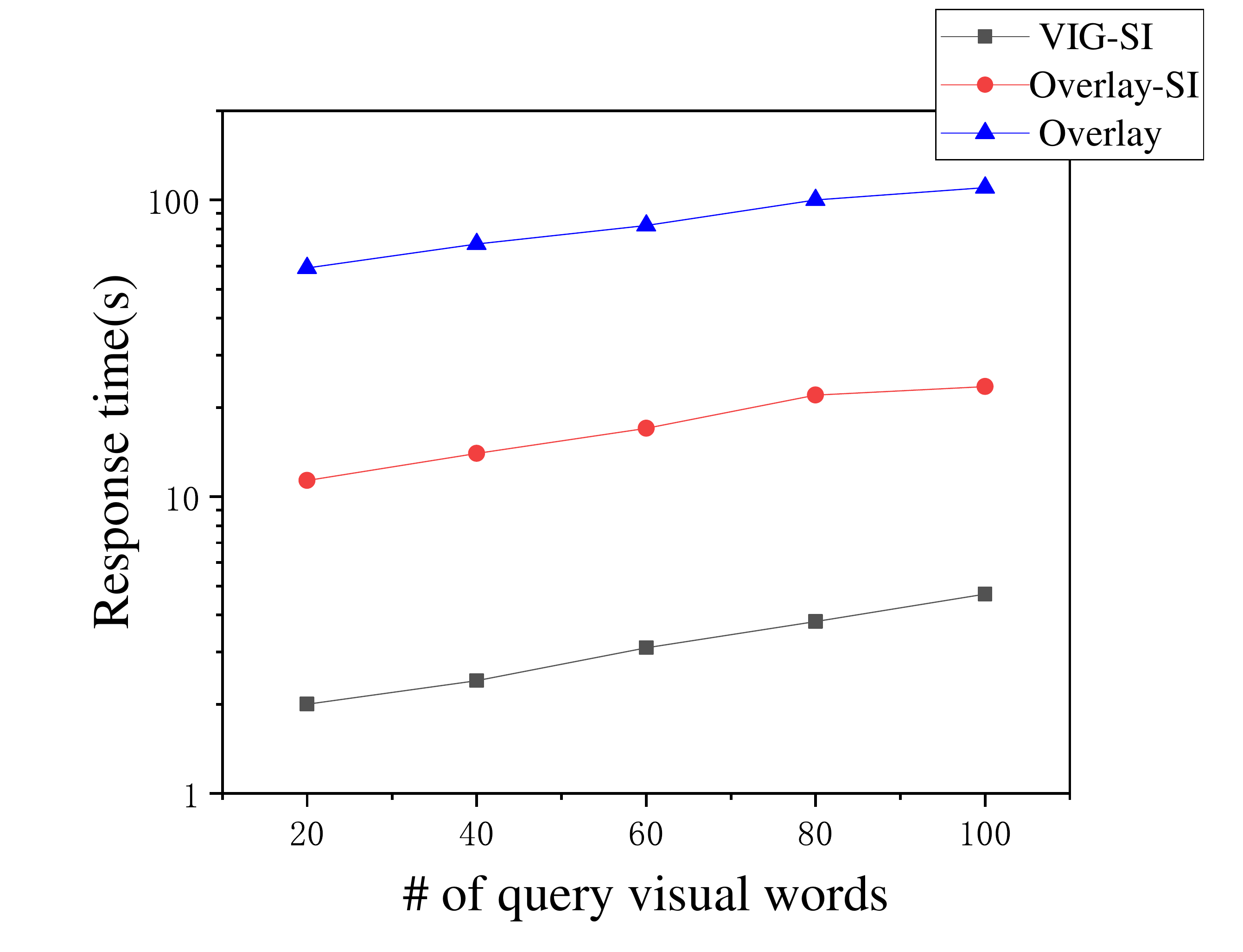}
     }
     \subfigure[Communication Time]{
     \includegraphics[width=0.48\linewidth]{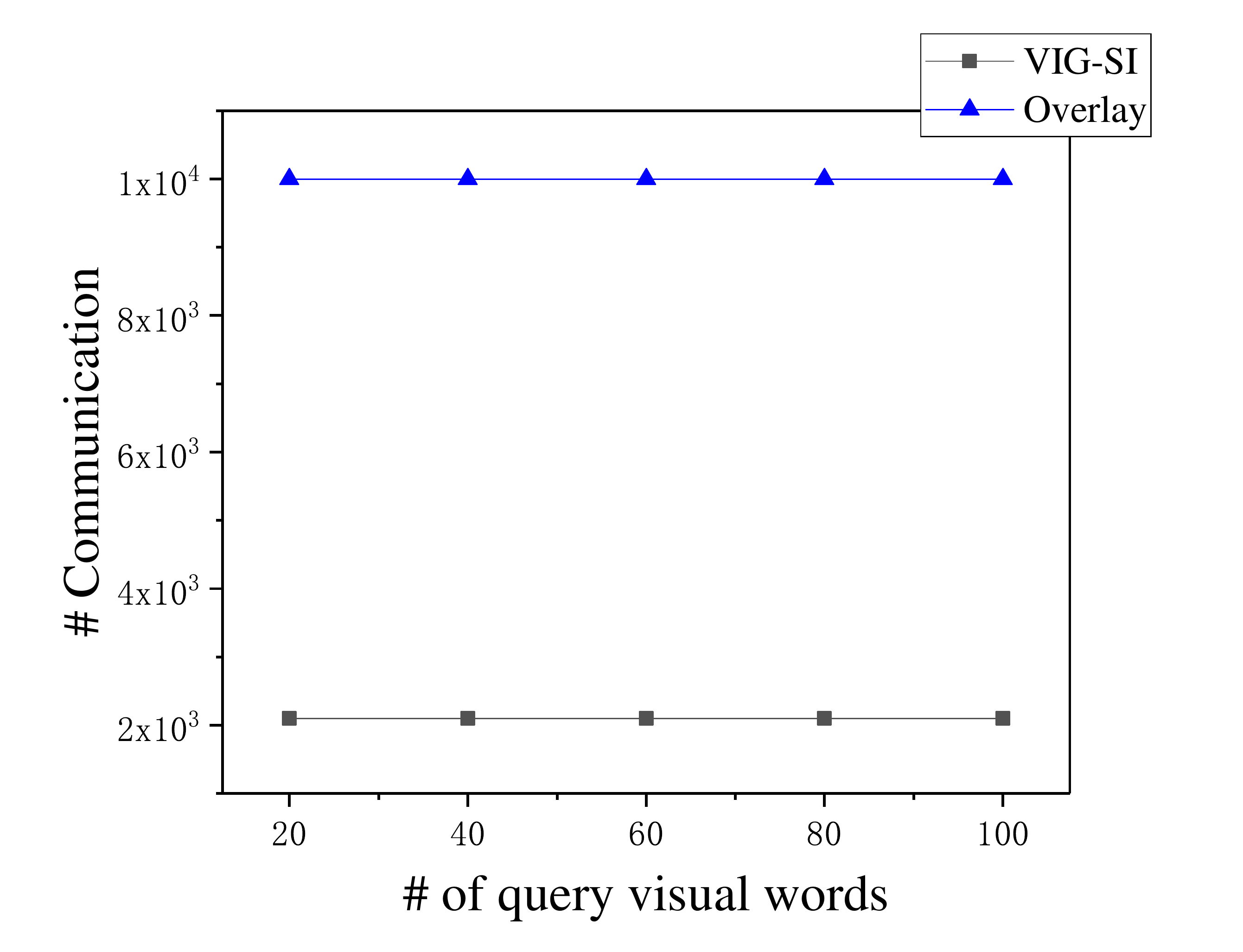}
     }
   \captionsetup{justification=centering}
       \vspace{-0.2cm}
\caption{Effect of query visual words}
\label{fig:query-visual}
\end{center}
\end{minipage}
\end{figure*}

\noindent \textbf{Evaluation on the number of results.}
Figure ~\ref{fig:query-result}
reports the performance of algorithms when the number of result $k$ varies from
10 to 50. Not surprisingly, all of them increase gradually with the
increase of $k$ and the performance of \textbf{VIG-SI} is always the best. Similar to the above situation, the communication cost of both \textbf{Overlay} and \textbf{VIG-SI} is still stable, and \textbf{VIG-SI} has better performance.

\begin{figure*}
\newskip\subfigtoppskip \subfigtopskip = -0.1cm
\begin{minipage}[b]{1\linewidth}
\begin{center}
     \subfigure[Time]{
     \includegraphics[width=0.48\linewidth]{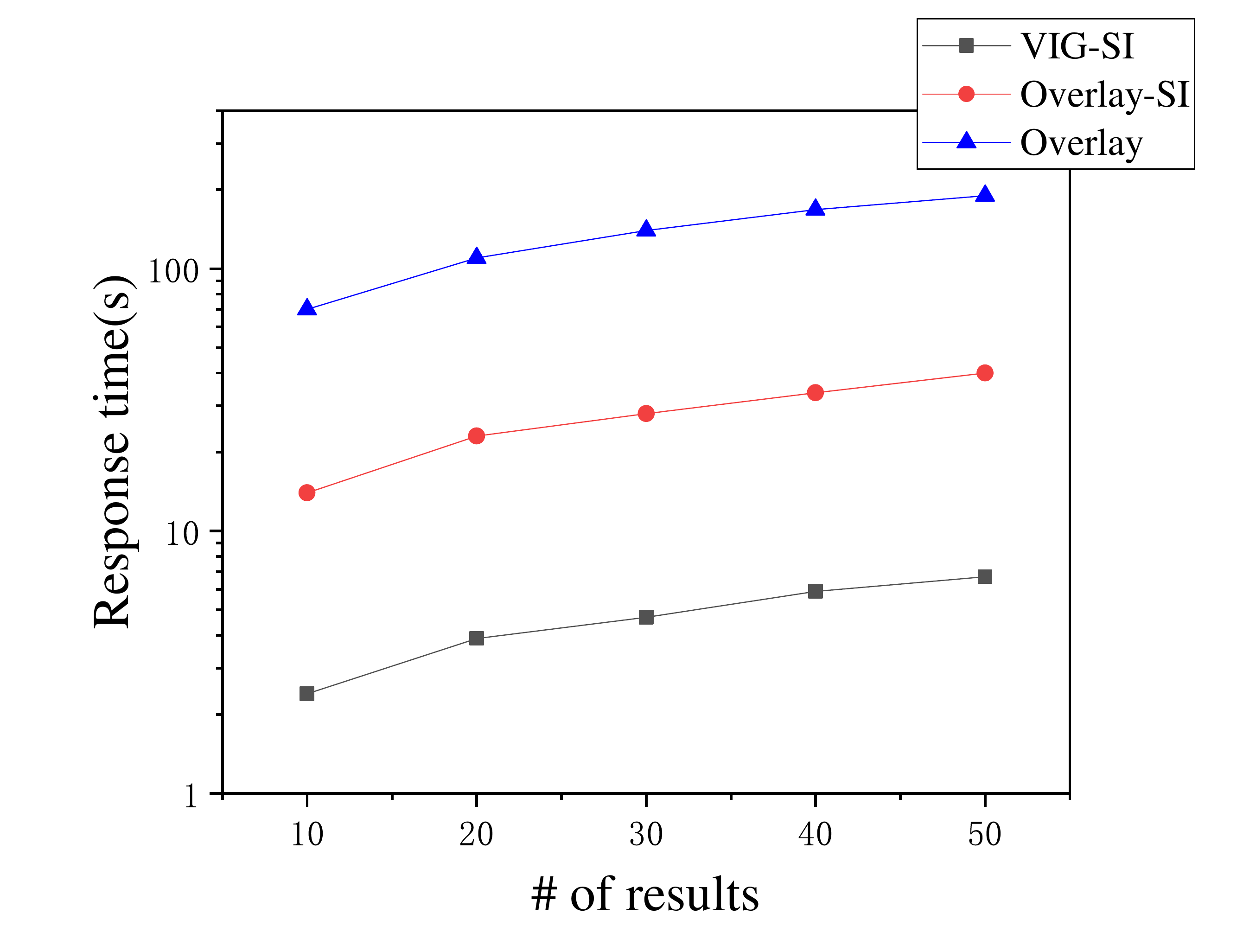}
     }
     \subfigure[Communication Time]{
     \includegraphics[width=0.48\linewidth]{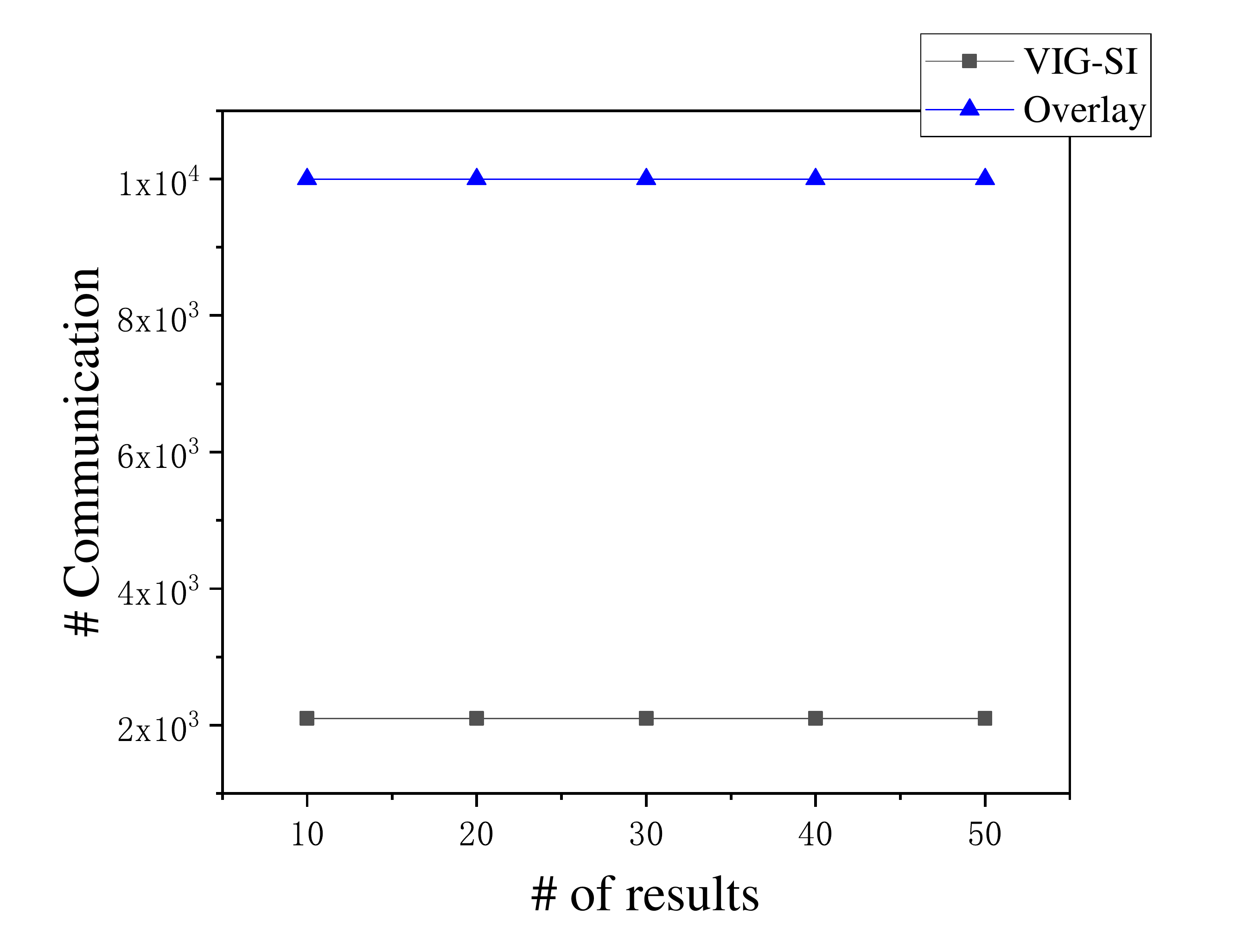}
     }
   \captionsetup{justification=centering}
       \vspace{-0.2cm}
\caption{Effect of results}
\label{fig:query-result}
\end{center}
\end{minipage}
\end{figure*}

\noindent \textbf{Evaluation on the preference parameter.}
The effect of query preference parameter $\mu$ is shown in Figure ~\ref{fig:preference-parameter}. A large value of $\mu$ indicates high preference to network proximity, while a small value of $\mu$ means high preference to visual similarity. As can be seen in Figure ~\ref{fig:preference-parameter}(a), all algorithm are sensitive to $\mu$, and their performance becomes better as the increase of $\mu$.
Figure ~\ref{fig:preference-parameter}(b) shows that the communication time is insensitive to $\mu$.

\begin{figure*}
\newskip\subfigtoppskip \subfigtopskip = -0.1cm
\begin{minipage}[b]{1\linewidth}
\begin{center}
     \subfigure[Time]{
     \includegraphics[width=0.48\linewidth]{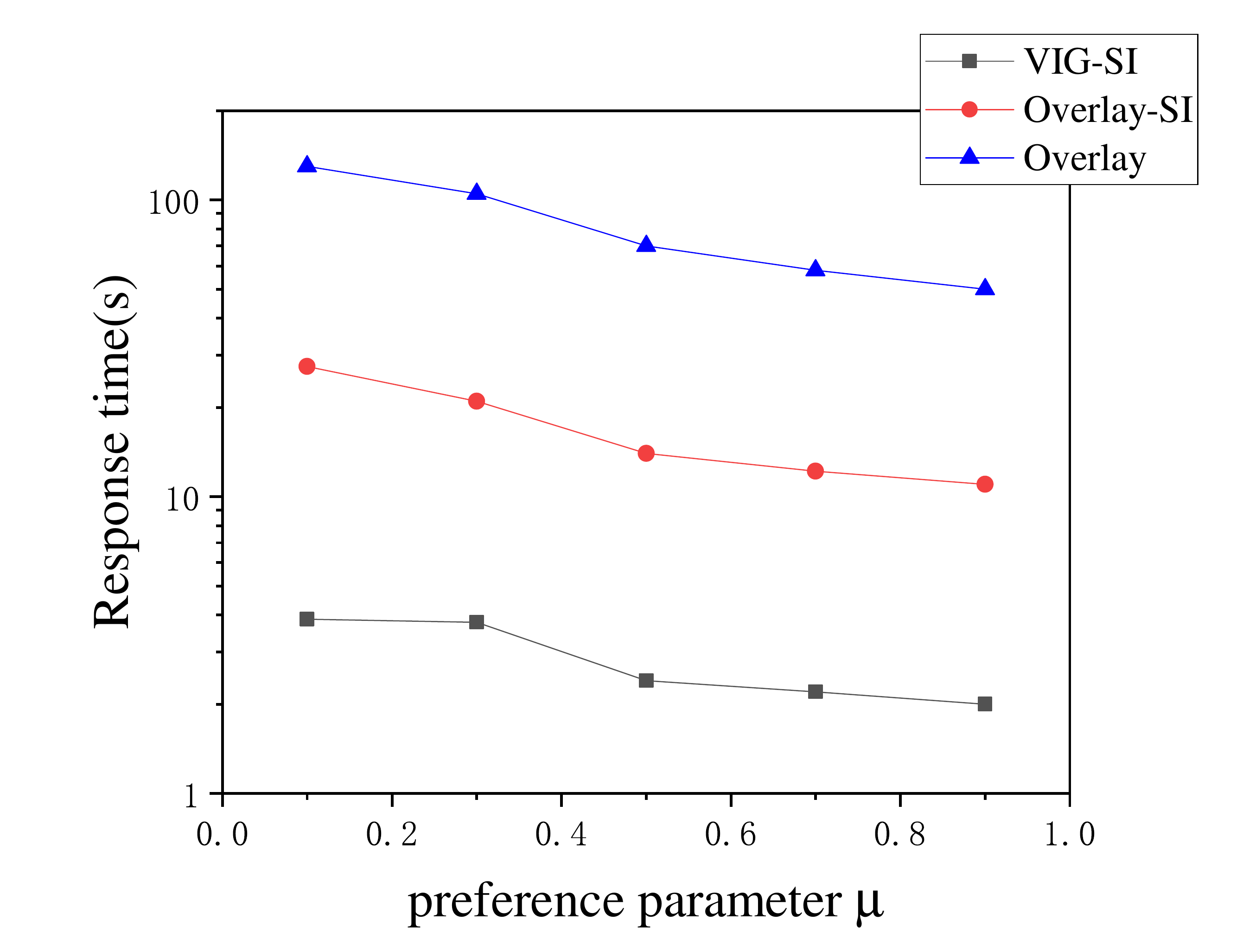}
     }
     \subfigure[Communication Time]{
     \includegraphics[width=0.48\linewidth]{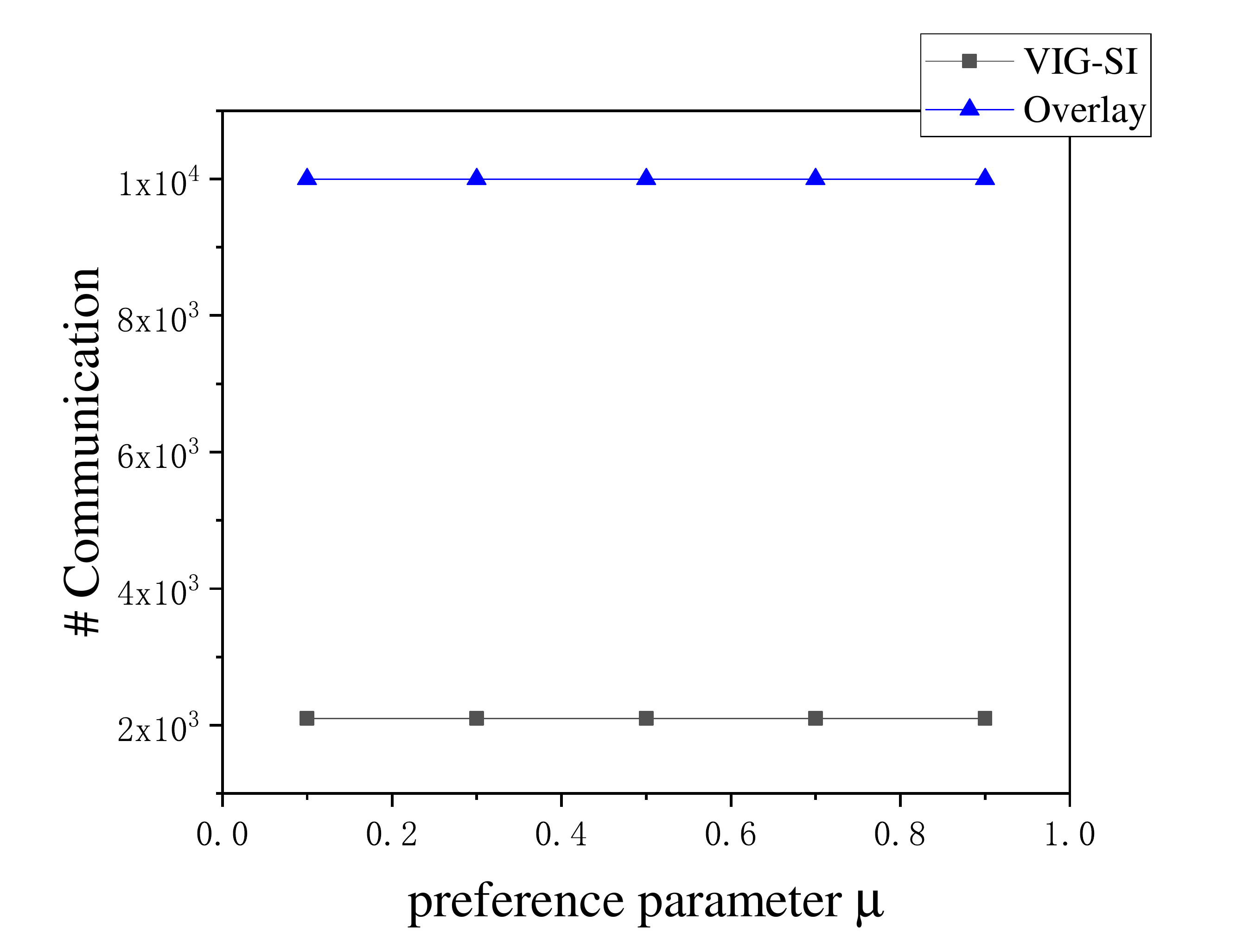}
     }
   \captionsetup{justification=centering}
       \vspace{-0.2cm}
\caption{Effect of preference parameter}
\label{fig:preference-parameter}
\end{center}
\end{minipage}
\end{figure*}

\noindent \textbf{Evaluation on the dataset size.}
Figure ~\ref{fig:dataset-size}(a) illustrates
the query response time of the under different sizes of image datasets. It is shown that all algorithms are sensitive to the growth of dataset size. As expected, Figure ~\ref{fig:dataset-size}(b) shows that the number of communication time of both algorithms are still stable when dataset size varies from 400K to 2M.

\begin{figure*}
\newskip\subfigtoppskip \subfigtopskip = -0.1cm
\begin{minipage}[b]{1\linewidth}
\begin{center}
     \subfigure[Time]{
     \includegraphics[width=0.48\linewidth]{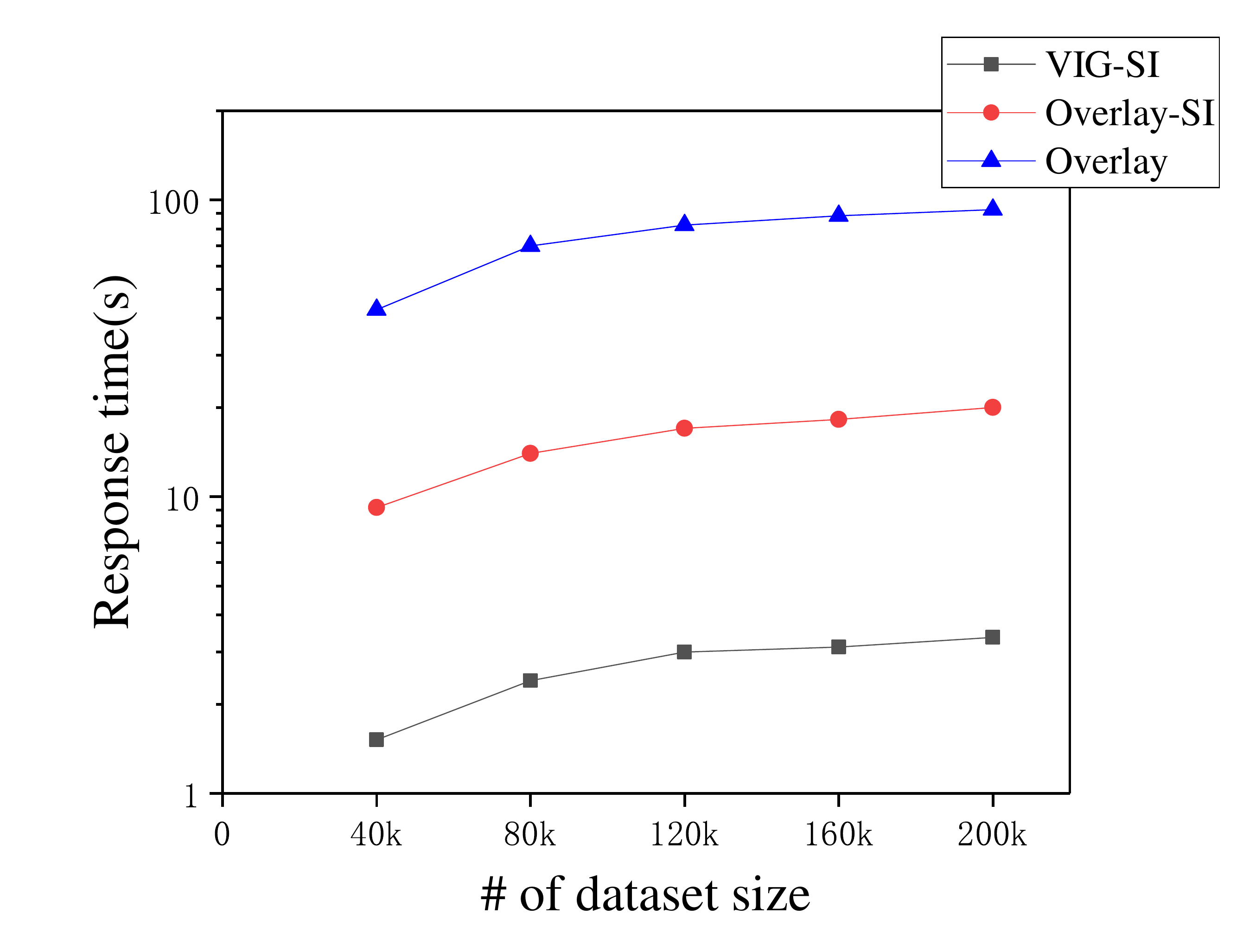}
     }
     \subfigure[Communication Time]{
     \includegraphics[width=0.48\linewidth]{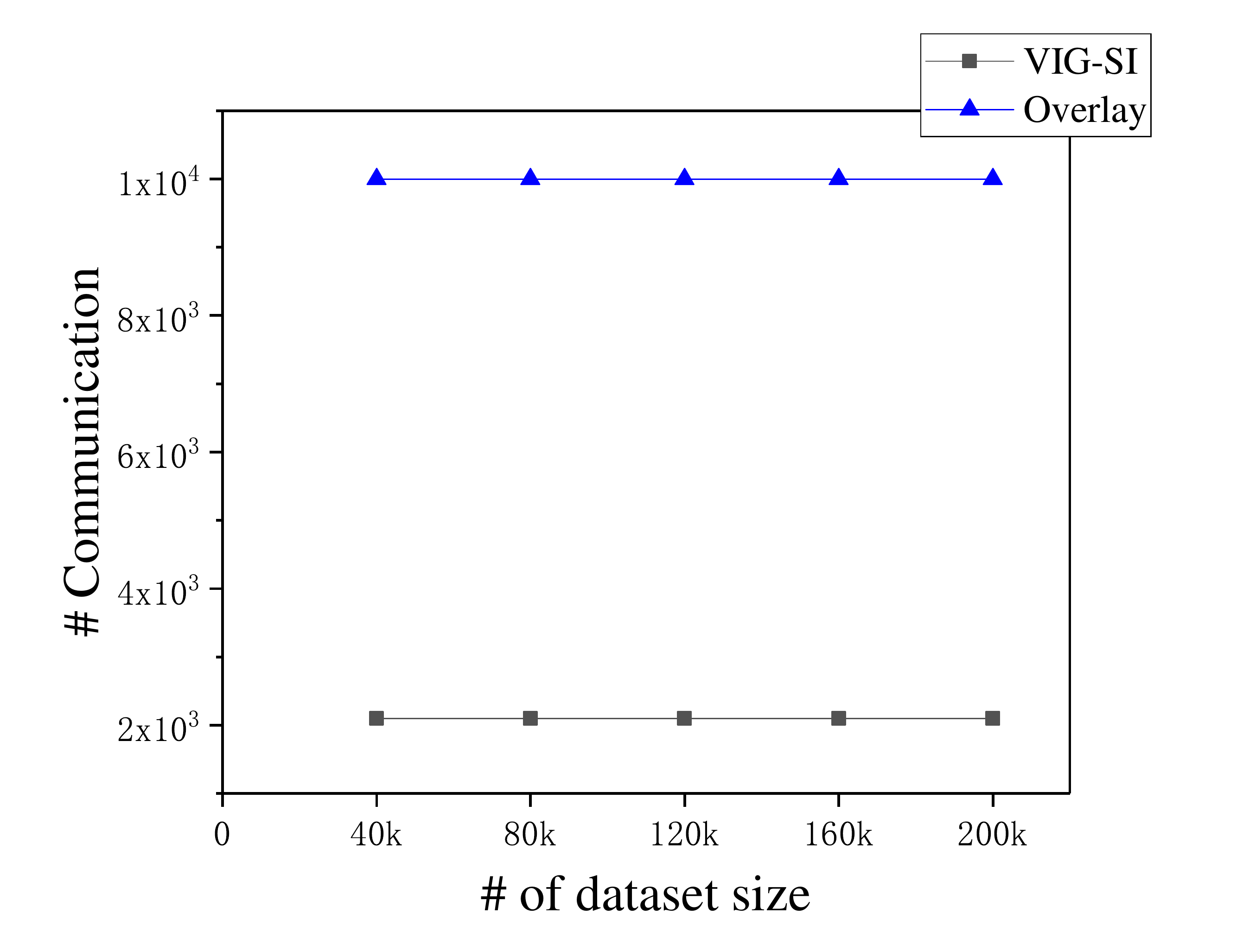}
     }
   \captionsetup{justification=centering}
       \vspace{-0.2cm}
\caption{Effect of dataset size}
\label{fig:dataset-size}
\end{center}
\end{minipage}
\end{figure*}

\noindent \textbf{Evaluation on different dataset.}
We evaluate the response time and the number of communication time for three real road network datasets of different size in Figure ~\ref{fig:different-dataset}. As can be seen in Figure ~\ref{fig:different-dataset}(a), all algorithms achieve the best performance in dataset NA, and have the worst performance in dataset AU. This is because more edges should be accessed as the edges size increases. The algorithm \textbf{Overlay} has better performance in dataset AU, while the algorithm \textbf{Overlay-SI} and \textbf{VIG-SI} achieve better performance in dataset NA, because the probability to invoke geo-visual search become higher, as the edge number increases. Figure ~\ref{fig:different-dataset}(b) illustrates \textbf{Overlay} gains a constant value since it reports invoke the geo-visual search for every location, while the communication cost of \textbf{Overlay-SI} and \textbf{VIG-SI} increase, as the edge number increases.

\begin{figure*}
\newskip\subfigtoppskip \subfigtopskip = -0.1cm
\begin{minipage}[b]{1\linewidth}
\begin{center}
     \subfigure[Time]{
     \includegraphics[width=0.48\linewidth]{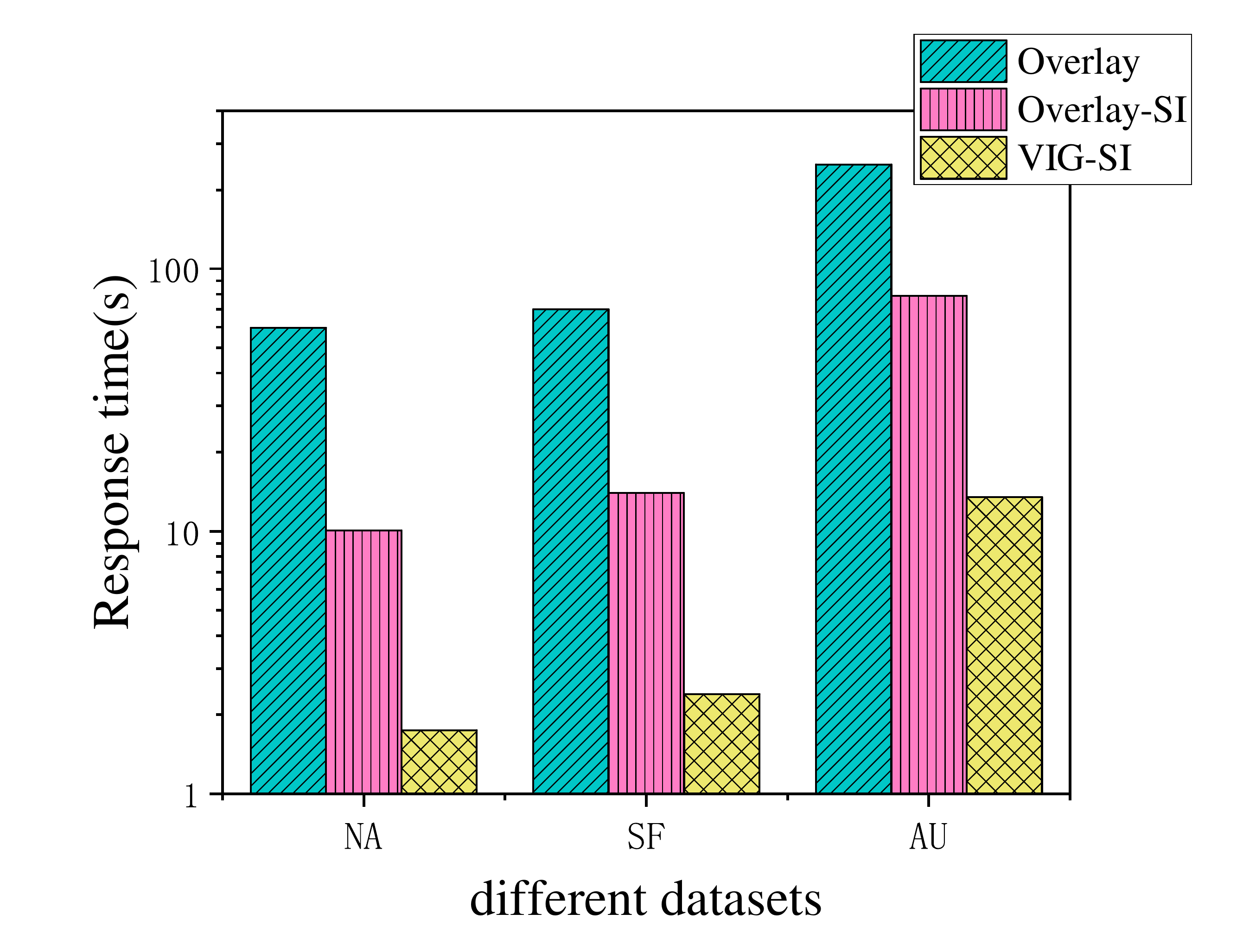}
     }
     \subfigure[Communication Time]{
     \includegraphics[width=0.48\linewidth]{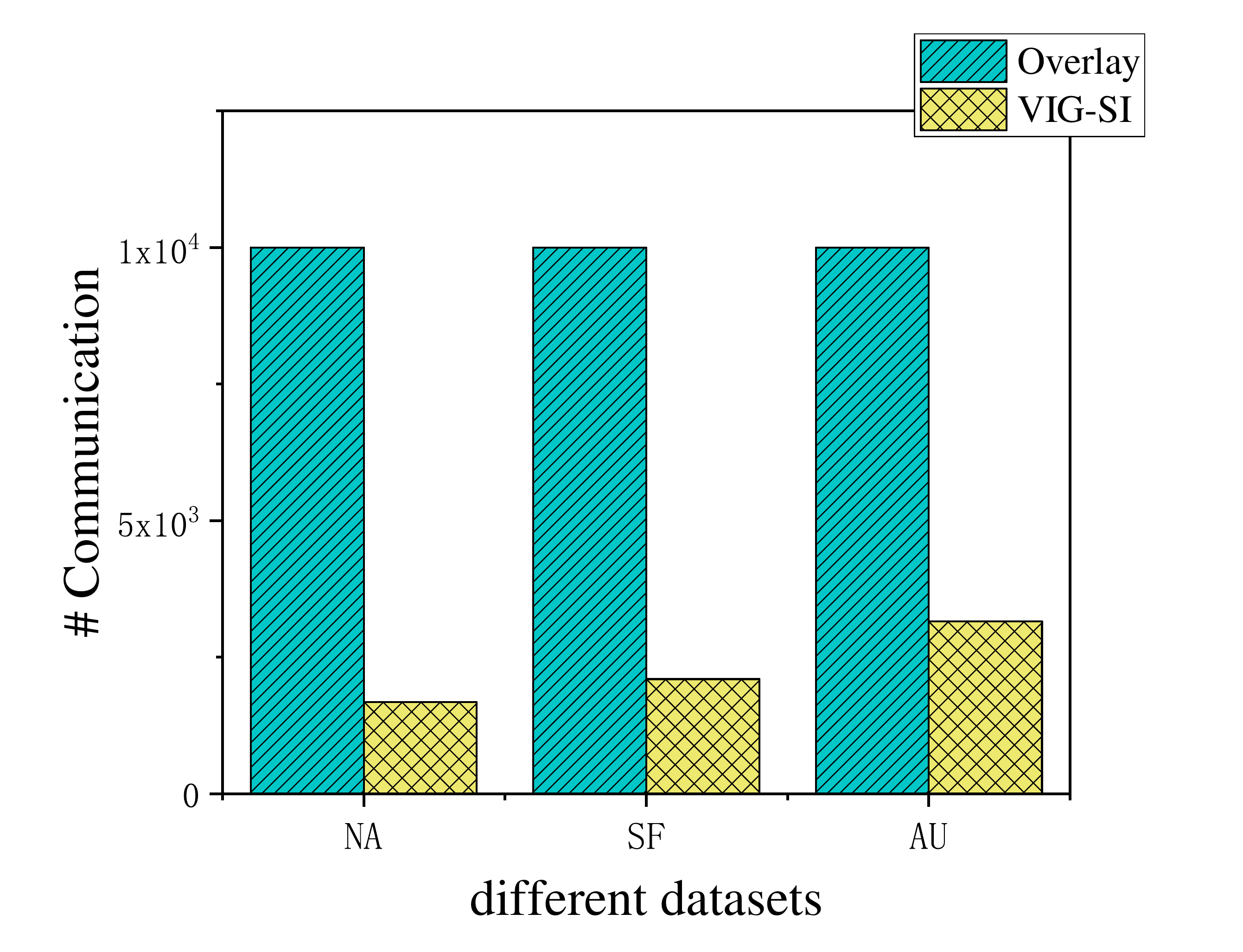}
     }
   \captionsetup{justification=centering}
       \vspace{-0.2cm}
\caption{Effect of different dataset}
\label{fig:different-dataset}
\end{center}
\end{minipage}
\end{figure*} 
\section{Conclusion}
\label{Conclusion}

In this paper, we propose and study a novel query problem named continuous top-$k$ geo-visual objects query (CT$k$GVOQ) on road network. Given a geo-visual objects which contains geographical information and visual content information, a CT$k$GVOQ aims to search out $k$ best geo-visual objects ranked in terms of visual content similarity or relevance to the query and the road network distance proximity to the query location. Firstly we define CT$k$GVOQ formally and propose the score function. In order to improve the efficiency of searching, we present a novel hybrid indexing framework called VIG-Tree and a efficient algorithm named geo-visual search on road network is proposed. To further reducing the computational cost in the process of query moving and searching our top-$k$ results from candidates set faster, we propose the notion of safe interval and introduce an efficient algorithm named moving monitor algorithm. The experimental evaluation on real multimedia dataset and road network dataset shows that our solution outperforms the state-of-the-art method.

\textbf{Acknowledgments:} This work was supported in part by the National Natural Science Foundation of China
(61702560), project (2018JJ3691, 2016JC2011) of Science and Technology Plan of Hunan Province, and the Research and Innovation Project of Central South University Graduate Students(2018zzts177). 



\bibliographystyle{spmpsci}      

\bibliography{ref}

\end{document}